\documentclass[aps,prx,twocolumn,superscriptaddress,raggedbottom,showpacs,floatfix,longbibliography]{revtex4-1}
\usepackage{amsmath,amsfonts,amssymb,amsthm,graphics}
\usepackage[next]{inputenc}
\usepackage[colorlinks=true,citecolor=blue,linkcolor=blue]{hyperref}
\usepackage{bbm}
\usepackage{bbold}
\usepackage{booktabs}
\usepackage{multirow}
\usepackage{hhline}
\usepackage{comment}
\usepackage{float}
\usepackage{latexsym}
\usepackage[dvipsnames]{xcolor}
\usepackage{graphicx}
\usepackage{epstopdf}
\usepackage{bm}
\usepackage{verbatim}
\usepackage{physics}
\usepackage{mathrsfs}
\usepackage{tikz}
\usepgflibrary{patterns} 
\usepgflibrary[patterns] 
\usetikzlibrary{patterns} 
\usetikzlibrary[patterns]
\usetikzlibrary{shapes,shadows,arrows,positioning,graphs}

\definecolor{TTH-color}{named}{green}
\definecolor{HH-color}{named}{magenta}
\definecolor{VG-color}{rgb}{0.97,0.57,0.11}

\definecolor{TTH-color2}{named}{red}
\definecolor{HH-color2}{named}{blue}
\definecolor{VG-color2}{rgb}{0.87,0.47,0.01}

\newtheorem{theorem}{Theorem}[section]
\newtheorem{corollary}{Corollary}[theorem]
\newtheorem{lemma}[theorem]{Lemma}

\begin{document}
\title{Many-Body Localization Landscape}
\author{Shankar Balasubramanian}
\affiliation{Department of Physics, Massachusetts Institute of Technology, Cambridge, MA 02139, USA}
\affiliation{Joint Quantum Institute, Department of Physics, University of Maryland, College Park, Maryland 20742-4111, USA}
\author{Yunxiang Liao}
\affiliation{Joint Quantum Institute and Condensed Matter Theory Center, Department of Physics, University of Maryland, College Park, Maryland 20742-4111, USA}
\author{Victor Galitski}
\affiliation{Joint Quantum Institute and Condensed Matter Theory Center, Department of Physics, University of Maryland, College Park, Maryland 20742-4111, USA}

\begin{abstract}
We generalize the notion of ``localization landscape,'' introduced by  M.  Filoche  and  S.  Mayboroda [Proc.  Natl.  Acad.  Sci. USA {\bf 109}, 14761 (2012)] for the single-particle Schr{\"o}dinger operator, to a wide class of interacting many-body Hamiltonians. The many-body localization landscape (MBLL) is defined on a graph in the Fock space, whose nodes represent the basis vectors in the Fock space and edges correspond to transitions between the nodes connected by the hopping term in the Hamiltonian. It is shown that in analogy to the single-particle case, the inverse MBLL plays the role of an effective potential in the Fock space. We construct a generalized discrete Agmon metric and prove Agmon inequalities on the Fock-state graph to obtain bounds on the exponential decay of the many-body wave-functions in the Fock space. The corresponding construction is motivated by the semiclassical WKB approximation, but the bounds are exact and fully quantum-mechanical. We then prove a series of locality theorems which establish where in the Fock space we expect eigenstates to localize.  Using these results as well as the locator expansion, we establish evidence for the existence of many-body localized states for a wide-class of lattice models \emph{in any physical dimension} in at least a part of their Hilbert space. The key to this argument is the observation that in sharp contrast to the conventional locator expansion for the Green's function, the locator expansion for the landscape function contains no resonances. For short-range hopping, which limits the connectivity of the Fock-state graph, the locator series is proven to be convergent and bounded by a simple geometric series. This, in combination with the discrete Agmon-like inequalities and the locality theorems, shows that localization for a fraction of the Hilbert space survives weak interactions and weak hopping at least for some realizations of disorder, but cannot prove or rule out localization of the entire Hilbert space. We qualitatively discuss potential breakdown of the locator expansion in the MBLL for long-range hopping and the appearance of a mobility edge in higher-dimensional theories. 

\end{abstract}

\maketitle

\section{Introduction \label{sec:introduction}}
There have been two major recent developments in the study of Anderson localization -- a mature field pioneered by Anderson~\cite{Anderson} back in 1958. First, the breakthrough paper by Basko, Aleiner, and Altshuler~\cite{BAA1,BAA2,BAA3} and  follow-up theoretical and experimental research have demonstrated that single-particle localization can survive interactions in many-particle disordered systems~\cite{Huse2007,Huse2010,Imbrie,MBL1D-0,MBL1D-1,MBL1D-2,MBL1D-3,MBL1D-4,MBL1D-5,MBLexp1D,MBLexp2D,DeMarco,Monroe,Review-Huse,Review-Serbyn}. This localization in the Fock space~\cite{Gornyi,BAA1} was called many-body localization (MBL), which remains a fast-developing field, with a number of open questions and controversies still unresolved. A second development came as a surprise from the mathematical analysis community, in which Svitlana Mayboroda and collaborators have developed a new, intuitive approach to single-particle Anderson localization~\cite{pnas,ll-dual,potential-1,potential-2}. Namely, it was proven that  a simple structure, $ u({\bf x})$ -- defined via $\left[-\nabla^2 + V({\bf x}) \right] u({\bf x}) = 1$ (here $V$ is the disorder potential) and called the localization landscape (LL)~\cite{pnas,ll-dual,ll-hop,ll-semiI,ll-semiII,ll-semiIII} -- represents an effective potential, where localized states up to a certain energy window  {\em actually} reside. More specifically, it was shown that in contrast to the bare disorder potential, whose explicit form does not necessarily shed much light on the structure of localization, the LL carries this information in a straightforward way.  Valleys of the inverse localization landscape (which includes both the effect of disorder and hopping/kinetic energy in a non-peturbative way) are ``traps'' that confine the quantum particle~\cite{pnas,ll-dual}. A series of bounds and estimates have been proven to show exponential localization of the wave-function in the LL ``traps.''~\cite{potential-1,potential-2,ll-exp} 

As far as the first major development goes -- regarding MBL -- there are currently four pieces of evidence for its existence. First, numerous exact diagonalization studies of finite-size systems -- both spin chains with random disorder and one-dimensional interacting fermions in a quasiperiodic potential (e.g., the interacting Aubry-Andr{\'e} model~\cite{AA1,AAH1,AAH2}) contain  signatures of a  transition~\cite{MBL1D-1,MBL1D-2,MBL1D-3}, or possibly crossover, from the localized, non-ergodic phase to a metallic, ergodic one. The numerical metrics used to characterize the many-body phases include calculating many-body inverse-participation ratio~\cite{IPR} and level statistics~\cite{Huse2010} or its proxies. However, these careful exact diagonalization studies can suffer from strong finite-size effects, which limit the usefulness of a numerical method in providing unambiguous evidence. Furthermore, recent and disconcerting data analysis by  Toma{\v{z}} Prosen and collaborators~\cite{Prosen} provided tentative evidence, based on data collapse in a model widely believed to be many-body localized, that favors ``many-body quantum chaos'' (i.e., an ergodic phase with Wigner-Dyson level statistics), although this interpretation remains controversial. While the exact diagonalization studies can neither serve as a conclusive proof of MBL as a stable phase~\cite{stability-1}, nor rule it out, they do seem to provide a faithful description of physics actually seen in cold-atom experiments (at least in one dimension). The experimental studies, in particular by Immanuel Bloch et al. (in both one and two dimensions)~\cite{MBLexp1D,MBLexp2D}, Brian de Marco et al.~\cite{DeMarco}, and Chris Monroe et al.~\cite{Monroe}, represent a second piece of evidence. However, given the subtle mathematical nature of the question regarding the existence of MBL and a number of additional complications in experiments (the finite lifetime of many-body states, presence of the trap potential, etc), experiments can only provide strong indication that MBL exists, but not a proof. 

The third piece of evidence is an actual proof by John Imbrie in one dimension~\cite{Imbrie}. This work is one of a few papers at a mathematical level of rigor addressing the MBL problem. However, it includes a physically reasonable but mathematically {\em ad hoc} assumption about the absence of level attraction. As was pointed out by Polkovnikov \cite{Polkovnikov}, coupling to a bath may in principle lead to level attraction. Indeed, there are examples in the literature where phenomena of this type have been reported \cite{levattraction1,levattraction2}. However, it is unclear at this stage whether interactions may lead to an effective bath that would enable level attraction. It appears unlikely, but until this issue is mathematically settled, Imbrie's proof is incomplete. Finally, the fourth piece evidence for MBL comes from the original MBL work by Basko, Aleiner, and Altshuler~\cite{BAA1,BAA2,BAA3}. The work is based on the locator expansion, which is also a method Anderson used in his pioneering paper~\cite{Anderson}. The key idea of the locator expansion is simple. One starts with a strongly disordered model, with hopping ``turned'' off (which is trivially localized), and performs a series expansion in the hopping to obtain the Green's functions. The method is straightforward, but suffers from divergences associated with resonances in the expansion, which does not necessarily converge. A renormalization scheme exists to resum the series, but it is not mathematically rigorous  (unless the hopping takes place on a Caylee graph and specifically a Bethe lattice, in which rigorous statements have been known since the early work by Abou-Chacra,  Anderson, and Thouless~\cite{Anderson2}).  All in all, there exists overwhelming evidence for the existence of MBL, but a conclusive mathematical proof for physically relevant lattices is still lacking. 

In contrast, the localization landscape of Mayboroda and collaborators~\cite{pnas,ll-dual,potential-1,potential-2} is a set of mathematically rigorous statements and bounds which physically correspond to single-particle Anderson localization. The key to the LL construction is a set of  universal Agmon-type bounds proven in Ref.~\cite{agmon-1,agmon-2}. Agmon's estimates have been known since 1979 and represent a set of {\em exact} inequalities for the eigenstates of the single-particle Schr{\"o}dinger operator
(which can be generalized to other operators). The construction of the estimates is motivated by a semiclassical WKB analysis, but Agmon's estimates hold beyond the usual WKB regime. Particular Agmon's bounds are potential-specific and have been developed for a wide variety of problems such as the double-well potential~\cite{well1,well2,well3}, a particle in a magnetic field~\cite{mag}, the Klein-Gordon equation~\cite{KG1,KG2}, etc. A remarkable achievement of Mayboroda et al~\cite{potential-1,potential-2,ll-exp} is that they have proven a set of universal Agmon's estimates for the eigenstates of the  Schr{\"o}dinger operator in a random potential (while potential does not have to be random for the construction to hold, the bounds appear to be most useful when it is), but with the role of an effective potential played by the inverse landscape function, $1/u({\bf x})$ rather than the bare potential itself. Most importantly these bounds predict exponential decay of eigenfunctions localized in the valleys of the LL. This statement is both mathematically rigorous and supported by numerical simulations, which provide an intuitive illustration of the structure of the spectrum. Note that the method works in an arbitrary dimension, and may contain a route to describing the delocalizating transition in three and higher dimensions via a percolation in the LL~\cite{MayborodaIASTalk}, although it has not been rigorously described yet. 

This paper brings together many-body localization and localization landscape by constructing a many-body localization landscape (MBLL) on a graph in Fock space. We prove that the MBLL  serves as an effective potential where many-body states are localized. First we define the MBLL for a wide class of interacting spin Hamiltonians on an arbitrary-dimensional lattice or more generally on a graph with nearest neighbor interactions (or equivalently hoppings), $t$. 
Next, we generalize Mayboroda's arguments to an undirected discrete graph in the Fock space. The vertices of the graph correspond to many-body states and the edges to transitions enabled by the hopping term in the Hamiltonian.  We show that the inverse many-body localization landscape, $1/u_\alpha$ can be viewed as an effective potential in the Fock space, by deriving the following equation 
\begin{equation}
    - t \sum_{\beta \in N(\alpha)} u_\beta u_\alpha \left(\frac{\psi_\beta}{u_\beta}-\frac{\psi_\alpha}{u_\alpha}\right) + \psi_\alpha = E' u_\alpha \psi_\alpha,
\end{equation}
which represents a discrete Fock-space generalization of Eq.~(5) in Ref. \cite{potential-1} by Mayboroda et al. involving a ``weighted'' Laplacian ($\alpha$ and $\beta$ represent nodes of the Fock-space graph corresponding to the many-body states $\ket{\alpha}$ and $\ket{\beta}$, $N(\alpha)$ is the set of neighbors of node $\alpha$, and $E'$ is a shifted many-body energy as explained in the text). The next step involves a  generalization of Agmon's bounds for this equation. It is accomplished by following semi-classical WKB-like intuition and a subsequent path-integral-like construction on the Fock-space graph. In analogy with the single-particle Agmon's bound, a many-body Agmon's bound can be proven, which is exact, and rigorously establishes (given certain properties of the landscape, as discussed in the text) exponential decay of the eigenfunctions in the valleys of the MBLL. 

It is important to note that since exponential decay occurs in the valleys of the MBLL (or the mountains of the inverse MBLL), the Agmon bounds do not rule out the possibility of the wavefunction spreading out across all of the MBLL mountains.  To better quantify this, we prove a series of locality theorems; these theorems roughly states that eigenstates in the global Fock-space graph are close to eigenstates defined in a particular well of the inverse MBLL if the corresponding energies are close to each other.  Therefore, eigenstates at a particular will primarily lie in the set of wells whose corresponding spectra resonate at that energy.  These can severely restrict the possible support of the eigenstate, assuming that we choose a particular realization of disorder.

While these exponential bounds are exact, they assume the existence of valleys in the MBLL, which is a priori not guaranteed. We then perturbatively explore the topology of the many-body localization landscape using an Anderson-like locator expansion. Importantly, to establish the existence of the valleys, we do not need information about the entire spectrum (which is provided by the Green's function in the usual locator expansion) and the knowledge of the MBLL suffices. Note that  finding the latter is much simpler, as it involves inverting a matrix (as opposed to fully diagonalizing it, which in effect is required for the full Green's function and is more computationally costly). Another critical observation is that the locator expansion for MBLL {\em does not have the problem of resonances}, which has been the major roadblock in studies of the localization physics using conventional methods. This simplification allows us to circumvent the need for renormalization in the locator expansion, bound it by a geometric series, and hence prove convergence for a wide class of models. If we start with a strongly disordered Hamiltonian with no hopping, localization at the bare level is trivial, and the convergence of the series in conjunction with the locality theorem, essentially proves that {\em some} states remain many-body localized for particular realizations of disorder. However, this analysis cannot provide useful bounds for all states in the Hilbert space. It is also unclear at this stage what ``fraction'' of the Hilbert space can be proven to remain localized per this method 
(this seems to remain an open problem for the single particle localization landscape as well). We discuss this critical point towards the end of the paper  along with the structure of the Fock space graph corresponding to physical lattices of different dimension. Simple estimates for graph connectivity (which enters the many-body Agmon bounds) point to a route that can help explore mobility-edge physics in higher-dimensional MBL models. 

\section{Model and The Fock-Space Graph}
The canonical model that is used to describe localization is the Anderson model, which consists of non-interacting particles in a random potential.  We are interested in understanding the effects of adding an interaction.  The simplest model which incorporates nontrivial interaction is given by the Hamiltonian:
\begin{equation}\label{eq:Ham}
    \mathcal{H} = -t\sum_{\langle i, j\rangle}(\sigma^+_i \sigma^-_{j} + \text{h.c.}) + \sum_i \epsilon_i n_i + V \sum_{\langle i, j\rangle} n_i n_j,
\end{equation}
where $\sigma^+$ and $\sigma^-$ are spin-$1/2$ raising and lowering operators, and $n = 2\sigma^z-1$, which we refer to as the ``density''. 

The first term denotes hopping between nearest neighbors on graph $\mathcal{G}_P$ with edge set $E$ and vertex set $V$.  For instance, this graph can be a one-dimensional lattice or a square lattice, but for the sake of generality we assume it can be arbitrary. We assume that $t\ge 0$; however, if the graph is bipartite, the sign of $t$ is unimportant and can always be made negative by the mapping $\sigma^+ \to -\sigma^+$ and $\sigma^- \to -\sigma^-$ on one of the subgraphs.  This latter case where $t\leq 0$ is more conventionally studied in the context of MBL.  Call $|V| = L$ the number of vertices in $\mathcal{G}_P$ (or correspondingly, the number of sites in the lattice).  

The second term denotes the disorder potential. In the case of the Anderson model, $\epsilon_j \sim U[0, W]$, where $W$ is a constant, but the specific distribution of disorder is unimportant for our construction.  The last term denotes density-density interactions that occur between nearest neighbors in $\mathcal{G}_P$; we may generalize the density-density interactions to quartic interactions, so long as the interaction strength is negative (although there exist graphs such that a similar particle-hole symmetry as described in the previous paragraph renders the sign unimportant).    

An interpretation that will be useful in the subsequent section involves the visualization of this Hamiltonian in Fock space.  Since the Hamiltonian in Equation \ref{eq:Ham} is particle conserving, we restrict the Hilbert space to the set of states with $N$ occupied sites.  A tuple of occupation numbers at each site uniquely describes a basis vector in the Hilbert space.  Associate each of these $M = \binom{L}{N}$ basis vectors to a vertex in a different graph $\mathcal{G}_F$.  Edges connecting two such vertices in $\mathcal{G}_F$ occur when the two configurations are connected by the hopping term.  The resulting ``hypergraph'' is essentially an embedding of the physical graph (or lattice) $\mathcal{G}_P$ in Fock space, and indicates the structure of the hopping.  This is a visual that allows us to better understand MBL in Fock space, as well as help develop intuition for bounds we present throughout the rest of the paper.  Figure \ref{fig:fock_exmp} is an example of the Fock-space graph $\mathcal{G}_F$ for a small system with $N=3$ particles; the physical graph $\mathcal{G}_P$ is a one-dimensional chain with $L=5$ sites and periodic boundary conditions.

Henceforth, we will use the Latin alphabet ($i$, $j$) to denote nodes in the real-space (``physical'') graph $\mathcal{G}_P$ and the Greek alphabet ($\alpha$, $\beta$) to denote nodes in the Fock-space graph $\mathcal{G}_F$.

\begin{figure}
    \centering
    \includegraphics[scale=0.42]{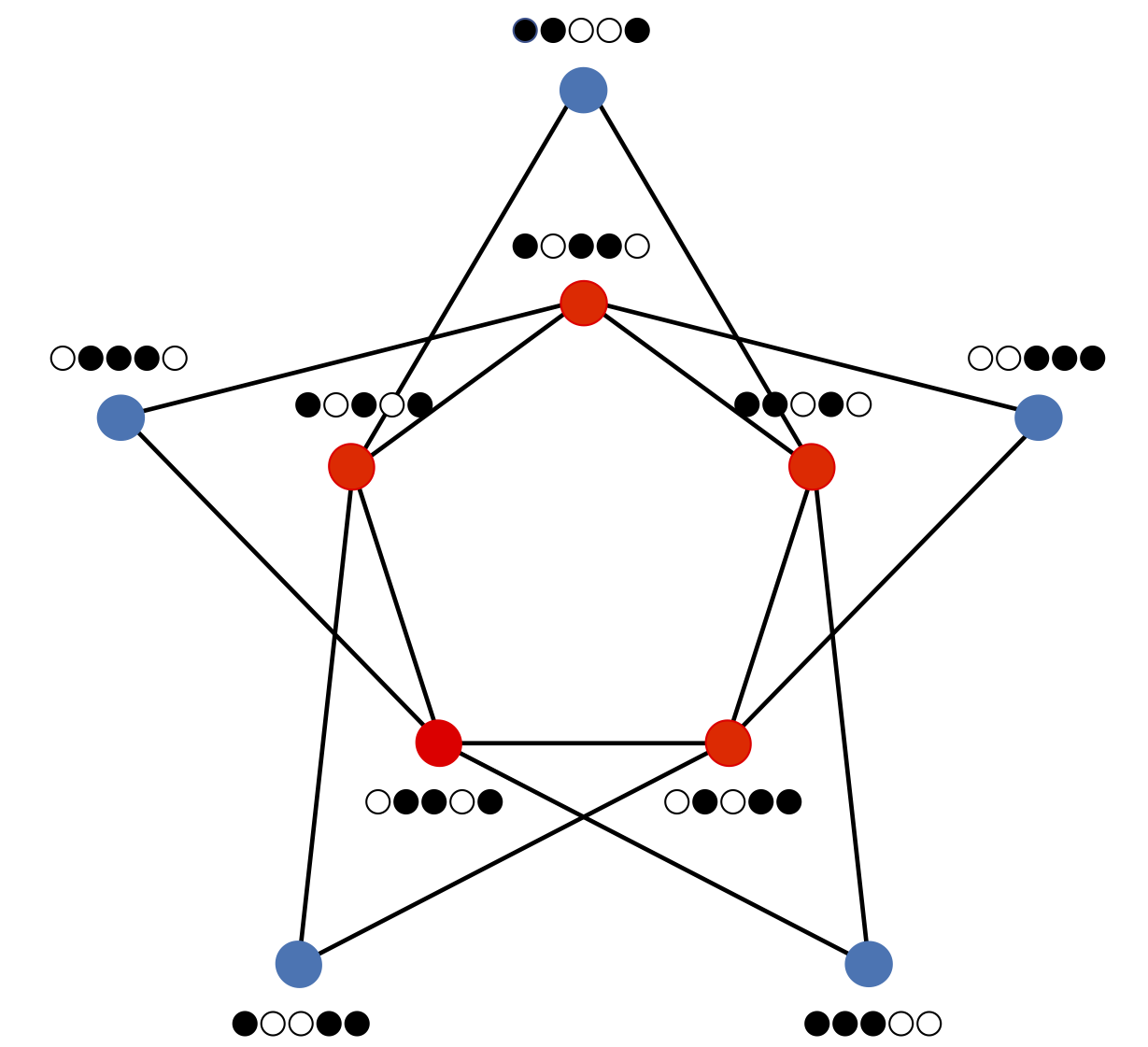}
    \caption{The Fock-space graph $\mathcal{G}_F$ for a one-dimensional chain with $L=5$ and $N=3$.  For convenience, the occupation of each state is denoted next to each node; blackened circles correspond to occupied sites and unfilled circles correspond to empty sites.  Notice that we assume periodic boundary conditions on the physical lattice in this construction, although this assumption is not needed in general.}
    \label{fig:fock_exmp}
\end{figure}

\section{Defining the MBLL}

In this section, we briefly visit the premiere results of Mayboroda et al.~\cite{pnas,ll-exp,potential-1,potential-2} regarding single particle localization via uniformly upper-bounding the eigenstates by the so-called landscape function.  We then show that the landscape formalism naturally extends to the many-body case with interactions.  This allows us to conclude similar bounds to those derived by Mayboroda et al.
\subsection{One-Particle Landscape}
Mayboroda et al.~\cite{pnas} considered the Hamiltonian $\mathcal{H} = -\nabla^2 + V(\bm{x})$ to solve the eigenvalue problem $\mathcal{H} \psi(\bm{x}) = E \psi(\bm{x})$, $\bm{x} \in \Omega \subset \mathbb{R}^d$ and $\psi(\bm{x})\Bigr|_{\partial \Omega}=0$.  Rather than solving the above Schr{\"o}dinger equation, the crucial insight is to solve $\mathcal{H} u(\bm{x}) = 1$, where $u(\bm{x})$ is the landscape function.  The resulting eigenstate satisfies $|\psi(\bm{x})| \leq E u(\bm{x}) \sup_{\bm{x}} |\psi(\bm{x})|$.  Thus, the landscape function can indicate whether states are localized or extended.  Moreover, it was also argued that $1/u(\bm{x})$ acts as an effective confining potential~\cite{potential-1,potential-2}, so at a given energy the eigenfunction can exponentially decay in classically forbidden regions.  Our goal is to extend this intuition to many-body systems.

Along with continuum models, discrete single particle systems were also analyzed.  Here, the wavefunction is discretized into a vector of scalar entries, the Laplacian term of the Hamiltonian is discretized so that $(\nabla^2) \psi_i = \psi_{i+1} + \psi_{i-1} - 2 \psi_i$, and the potential is treated as an onsite term. It was shown that solving the matrix equation $\mathcal{H} u = 1$ gives a landscape function that satisfies the same bounds as in the continuous case~\cite{ll-dual}.

\subsection{Many-body Landscape}

We now construct a similar landscape function for the Hamiltonian in Equation \ref{eq:Ham}.  To do so, we notice that in the Fock space representation, the Hamiltonian consists of a sum of two parts.  The first part is the diagonal contribution due to the onsite and interaction terms.  The second part is $-t$ times the adjacency matrix of graph $\mathcal{G}_F$.  Thus, this Hamiltonian behaves similarly to the single particle case albeit with long range hopping.  It thus suffices to show that a landscape function can be defined given the Hamiltonian in the Fock space.

As before, restrict the Hilbert space to states with a fixed number of occupied sites; label the corresponding basis states by $|\alpha\rangle$ (occupation number representation).  For an arbitrary superposition of such states, call $v_\alpha$ the coefficient of state $|\alpha \rangle$.  Then, the Hamiltonian acting on state $|v\rangle = \sum_{\alpha} v_\alpha |\alpha\rangle$ gives
\begin{equation}
    \langle \alpha| \mathcal{H} |v\rangle = -t \sum_{\beta \in N(\alpha)} v_{\beta} + \left(\sum_{i = 1}^{L} V n_i^{(\alpha)} n_{i+1}^{(\alpha)} + \epsilon_i n_i^{(\alpha)} \right)v_\alpha,
\end{equation}
where $n_i^{(\alpha)}$ denotes the occupation number of site $i$ for state $\alpha$, and $N(\alpha)$ is the set of neighbors of $\alpha$ in graph $\mathcal{G}_F$.  For the purposes of the following proofs, we also need to impose Dirichlet zero boundary conditions in the sense of the graph $\mathcal{G}_F$.  This can be accomplished by augmenting $\mathcal{G}_F$ to $\mathcal{G}_F'$ and connecting an additional node $\alpha'$ to each node $\alpha$.  We enforce the condition $v_{\alpha'} = 0$, which results in the above equation not changing.  At the end, we perform a restriction of $\mathcal{G}_F'$ to $\mathcal{G}_F$ without changing any of the main results.

The following Lemma gives conditions for when we should expect $v_{\alpha}$ to be positive.  The proof method is analogous to that of the maximum principle for harmonic functions in a bounded domain.
\begin{lemma}
Suppose $\langle \alpha| \mathcal{H} + K |v\rangle \geq 0$ where $K$ is a positive constant for all $\alpha$ and $v_{\alpha'} = 0$ for all $\alpha'$.  Then, if 
\begin{equation}
K - t |N(\alpha)| + \sum_{\langle i, j\rangle} V n_i^{(\alpha)} n_j^{(\alpha)} + \sum_{i} \epsilon_i n_i^{(\alpha)} \geq 0
\end{equation}
for all $\alpha$, we have that $v_{\alpha} \geq 0$.  Furthermore, if there exists a $\beta$ for which $\langle \beta| \mathcal{H} + K |v\rangle > 0$, then $v_\alpha > 0$ for all $\alpha$.
\end{lemma}
\begin{proof}
Suppose that, for the sake of contradiction, the global minimum $\text{argmin}_{\alpha, \alpha'} v_\alpha$ does not occur within the set of $\alpha'$ nodes.  With an abuse of notation, call $v_\alpha$ the global minimum.  Then, $v_\beta \geq v_\alpha$, where $\beta \in N(\alpha)$.  Therefore,
\begin{equation}
    \sum_{\beta \in N(\alpha)} \left(v_\beta - v_\alpha\right) \geq 0. \label{eq:localmin}
\end{equation}
This condition enforces that $v_\alpha$ is a local minimum with respect to its neighboring nodes in $\mathcal{G}_F'$.  Next, since $\langle \alpha| \mathcal{H} + K |v\rangle \geq 0$, it follows that
\begin{equation}
    t \sum_{\beta \in N(\alpha)} v_\beta \leq K v_\alpha + \left(\sum_{\langle i, j\rangle} V n_i^{(\alpha)} n_j^{(\alpha)} + \sum_{i} \epsilon_i n_i^{(\alpha)} \right)v_\alpha.\label{eq:ineqlemm1}
\end{equation}
Thus, combining with Equation \ref{eq:localmin}, we find
\begin{align}
    0 &\leq t \sum_{\beta \in N(\alpha)} \left(v_\beta - v_\alpha\right),\\ &\leq K v_\alpha - t |N(\alpha)| v_\alpha + \left(\sum_{\langle i, j\rangle} V n_i^{(\alpha)} n_j^{(\alpha)} + \sum_{i} \epsilon_i n_i^{(\alpha)} \right)v_\alpha.
\end{align}
To remind, note that $|N(\alpha)|$ is the number of neighbors of node $\alpha$ in $\mathcal{G}_F'$, which is one greater than the number of neighbors in $\mathcal{G}_F$.  Simple rearrangement gives
\begin{equation}
    0 \leq \left(K - t |N(\alpha)| + \sum_{\langle i, j\rangle} V n_i^{(\alpha)} n_j^{(\alpha)} + \sum_{i} \epsilon_i n_i^{(\alpha)}\right)v_\alpha.
\end{equation}
The quantity in the parenthesis is always positive by assumption, so $v_\alpha \geq 0$.  However, this contradicts the assumption that $v_\alpha$ is a global minimum, since it is at least equal to the value of $v_{\alpha'}$ for all $\alpha'$.  Thus, $v_\alpha$ only acquires a minimum in the nodes $\alpha'$, and it follows that $v_\alpha \geq 0$ for all nodes in $\mathcal{G}_F'$.

Proving strict positivity follows directly.  Suppose $v_\alpha = 0$ for some $\alpha$.  Then, reading from Equation \ref{eq:ineqlemm1}, we find
\begin{equation}
    t \sum_{\beta \in N(\alpha)} v_\beta \leq K v_\alpha + \left(\sum_{\langle i, j\rangle} V n_i^{(\alpha)} n_j^{(\alpha)} + \sum_{i} \epsilon_i n_i^{(\alpha)}\right) v_\alpha = 0,
\end{equation}
Since we proved that $v_\alpha \geq 0$ for all $\alpha$, we have that $v_\beta = 0$ for $\beta \in N(\alpha)$. We can represent this as the diffusion of an action on $\mathcal{G}_F'$: we set $v_\alpha$ equal to zero in node $\alpha$ and spread this action to all neighboring nodes.  We then inductively repeat.  Since $\mathcal{G}_F'$ is irreducible, there exists an iteration in which all nodes will equal zero -- thus all $v_\alpha = 0$.  Therefore, $\langle \alpha| \mathcal{H} + K |v\rangle  = 0$.  Hence, taking the contrapositive of this logic combined with $v_\alpha \geq 0$ finishes the proof.
\end{proof}
Notice that the constant $K$ equals
\begin{equation}
    K = - \min_{\alpha} \left[\sum_{\langle i, j\rangle} V n_i^{(\alpha)} n_j^{(\alpha)} + \sum_{i} \epsilon_i n_i^{(\alpha)} - t |N(\alpha)|\right].\label{eq:Kdef}
\end{equation}
This corresponds to a constant shift of the Hamiltonian, which does not affect the eigenstates.  Note that Mayboroda et al.~\cite{ll-dual} had to perform a similar shift to impose positive onsite potentials.  Next, we show that this conclusion implies that the elements of $(\mathcal{H}+K)^{-1}$ are positive:
\begin{lemma}
Consider the restriction of matrix $\mathcal{H}+K I \triangleq \mathcal{H}+K$ to states in $\mathcal{G}_F$, where $I$ denotes the identity matrix.  The inverse of the restriction of $\mathcal{H} + K$ is entry-wise positive.
\end{lemma}
\begin{proof}
Suppose $\langle \alpha| \mathcal{H} + K |v\rangle = f_\alpha$ for $f_\alpha \geq 0$.  By Lemma 1, $v_\alpha \geq 0$.  Inverting this equation and solving for $v_\alpha$ gives:
\begin{equation}
    v_\alpha = \sum_{\beta} \left(\mathcal{H} + K\right)^{-1}_{\alpha \beta}f_\beta. 
\end{equation}
Calling $f_\beta = \delta_{\gamma \beta}$, then $0 \leq v_\alpha = \left(\mathcal{H} + K\right)^{-1}_{\alpha \gamma}$.  Repeating for all $\alpha$ and $\gamma$ finishes the proof.  Strict positivity follows from the fact that the above choice of $f$ satisfies $f_\alpha > 0$ for at least one value of $\alpha$.
\end{proof}
In this next Lemma, we show that the spectrum of $\mathcal{H}+K$ is strictly positive.
\begin{lemma}
The matrix $\mathcal{H} + K$ has positive eigenvalues.
\end{lemma}
\begin{proof}
From Lemma 1, if $(\mathcal{H} + K) |v\rangle = 0$, then $v_\alpha \geq 0$.  Suppose the null space of $\mathcal{H} + K$ is not trivial; then it must contain some state $|v\rangle \neq 0$.  Thus, both $|v\rangle$ and $-|v\rangle$ are in the null space, which contradicts Lemma 1 since $v_\alpha$ cannot be negative; it follows that the null space of $\mathcal{H} + K$ is trivial.  This implies that there is no eigenvalue equal to zero.

It is trivial to show that Lemma 1 is satisfied for any choice of $K' > K$, and $\mathcal{H} + K'$ does not have any zero eigenvalues for $K' > K$.  Suppose, for the sake of contradiction, that there exists a negative eigenvalue $-\lambda$ of $\mathcal{H} + K'$.  Then, setting $K'' = K' + \lambda \geq K'$, it is clear that $\mathcal{H} + K''$ has a zero eigenvalue.  This contradicts the fact that $\mathcal{H} + K'$ does not have any zero eigenvalues for $K' > K$.  This completes the proof.
\end{proof}
We next define the MBLL to be the solution to 
\begin{equation}
    (\mathcal{H}+K) |u\rangle = |1\rangle,
\end{equation}
where $|1\rangle$ denotes the vector of all ones.  In an identical manner to the proof in Mayboroda et al.~\cite{ll-dual}, we arrive at the following bound on the eigenstates via the landscape function:
\begin{theorem}
Assume $(\mathcal{H}+K)|\psi\rangle = (E+K)|\psi\rangle$ with $\psi_{\alpha'} = 0$, where $E$ is the many-body energy, and $K$ is defined in Equation \ref{eq:Kdef}.  It follows that for all $\alpha$,
\begin{equation}
    \frac{|\psi_\alpha|}{\max_\beta |\psi_\beta|} \leq (E + K)u_\alpha.
\end{equation}
\end{theorem}
\begin{proof}
We use the two Lemmas that have been derived to write:
\begin{align}
    |\psi_\alpha| &= (E+K) \left |\sum_{\beta} \left(\mathcal{H} + K\right)^{-1}_{\alpha \beta} \psi_\beta\right|,\nonumber\\
    &\leq (E+K) \max_\beta |\psi_\beta| \sum_{\beta} \left(\mathcal{H} + K\right)^{-1}_{\alpha \beta},\nonumber\\
    &= (E+K) \max_\beta |\psi_\beta| u_\alpha,
\end{align}
where in the first line we use Lemma 3, in the second line we use the fact that $\left(\mathcal{H} + K\right)^{-1}_{\alpha \beta}$ is positive (from Lemma 2), and in the third line we use the definition of $u_\alpha$.
\end{proof}
Surprisingly, Mayboroda's single particle bound carries over in the many-body case even in the presence of interactions.  However, the visualization of this bound is more complex.  Rather than viewing the landscape function on a standard plot to determine domains of localization, one must view the landscape function with respect to the Fock space graph $\mathcal{G}_F$.  Localization occurs in \emph{connected} components of $\mathcal{G}_F$ with large values of $u$ which are surrounded by valleys with small values of $u$.

Also notice that our bounds depend on the choice of $K$.  While $K$ is arbitrary in principle, it is lower bounded.  Choosing too large of $K$ results in landscape functions that are not good representatives of the localization regions for the eigenstates.

From this point on, we denote $E' = E+K$ for simplicity and readability.  We will also change notation by calling $N(\alpha)$ the set of neighbors of node $\alpha$ in $\mathcal{G}_F$.

\section{MBLL as an Effective Potential}

The bound derived in the last section gives evidence that the MBLL informs us about possible localization in the Fock space.  This is consistent with the general qualitative picture of many-body localization, as discussed by Basko, Aleiner, and Altshuler~\cite{BAA1,BAA2,BAA3}, who proposed studying MBL in terms of the localization of eigenstates in Fock space.  However, to understand whether eigenfunctions exponentially decay outside of localized regions, we must borrow intuition from the single particle case.  Here, Mayboroda et al.~\cite{potential-1,potential-2} showed that $1/u$ behaves like an effective potential that confines the eigenstates.  This effective potential is often more informative than considering the original on-site potential, as it incorporates both disorder and the kinetic energy  in a non-perturbative way and more accurately identifies regions where eigenstates  localize.  We want to understand whether this analogy extends to the MBLL, as this gives significant insight into when and where in the Fock space we would expect strong localization to occur.  The result we show is that the rate of decay of eigenfunctions depends on the connectivity of a particular portion of $\mathcal{G}_F$: if a node has fewer neighbors, the resulting states have large contiguous regions of occupied sites and will more strongly localize.

The method we use to analyze this decay was first developed by Shmuel Agmon~\cite{agmon-1,agmon-2} several decades ago.  It essentially formalizes a semi-classical WKB analysis of wavefunction tunneling to exact upper bounds on the magnitude of the wavefunction, which are known as \emph{Agmon inequalities}.

\subsection{Effective Schrodinger Equation via MBLL}
 We will first proceed to rewrite the eigenvalue equation in terms of the MBLL.  Let us start by noting that the eigenstates $|\psi\rangle = \sum_\alpha \psi_\alpha |\alpha\rangle$ of the many-body Hamiltonian satisfy
\begin{equation}
    -t \sum_{\beta \in N(\alpha)} \psi_\beta + \left(\sum_{\langle i, j\rangle} V n_i^{(\alpha)} n_j^{(\alpha)} + \sum_i \epsilon_i n_i^{(\alpha)} \right)\psi_\alpha = E \psi_\alpha.
\end{equation}
For brevity, we call the term in the parenthesis $V_\alpha$.  Next, we perform the substitution $\psi_\alpha = \Psi_\alpha u_\alpha$.  We may then write
\begin{equation}
    -t \sum_{\beta \in N(\alpha)} \Psi_\beta u_\beta + V_\alpha \Psi_\alpha u_\alpha = E \Psi_\alpha u_\alpha,
\end{equation}
which can be rearranged as
\begin{align}
    &\left[-t \sum_{\beta \in N(\alpha)} u_\beta + (V_\alpha + K) u_\alpha\right] \Psi_\alpha - \\ &t \sum_{\beta \in N(\alpha)} u_\beta(\Psi_\beta-\Psi_\alpha) = E' \Psi_\alpha u_\alpha.\nonumber
\end{align}
The quantity in the brackets equals 1 by definition of the MBLL, and if we rewrite the equation in terms of $\psi_\alpha$, we obtain the equation
\begin{equation}\label{eq:landschromassaged}
     - t \sum_{\beta \in N(\alpha)} u_\beta u_\alpha \left(\frac{\psi_\beta}{u_\beta}-\frac{\psi_\alpha}{u_\alpha}\right) + \psi_\alpha = E' u_\alpha \psi_\alpha.
\end{equation}
This is a more revealing form, as the first term is equal to $t$ times a weighted Laplacian matrix $L$ acting on the state $|\psi/u \rangle = \sum_\alpha (\psi_\alpha/u_\alpha) |\alpha\rangle$.  A Laplacian matrix with weighted edges $w_{\alpha \beta}$ has elements equal to:
\begin{equation}
L_{\alpha \beta} = \left\{\begin{array}{lr}
        -w_{\alpha \beta}, & \text{for } \alpha \sim \beta\\
        \sum_{\alpha \in N(\beta)} w_{\alpha \beta}, & \text{for } \alpha = \beta
        \end{array}\right.
\end{equation}
where $\alpha \sim \beta$ means that nodes $\alpha$ and $\beta$ are connected in $\mathcal{G}_F$.  This quantity is a discrete generalization of the differential operator for a graph whose distance metric is given by the weights $w_{\alpha \beta}$.  In the above equation, the weights of the Laplacian are $w_{\alpha \beta} = u_\alpha u_\beta$, and so are symmetric with respect to interchanging $\alpha$ and $\beta$.  Moreover, they are also positive as per the analysis of the previous section.  Write the LHS of Equation \ref{eq:landschromassaged} as $M |\psi\rangle$, where $M_{\alpha \beta} = t L_{\alpha \beta}/u_\beta + \delta_{\alpha \beta}$.  Next, we apply $\langle \psi/u|$ to the matrix version of the above equation:
\begin{align}\label{eq:effschr}
   \mel{\frac{\psi}{u}}{M}{\psi} &= t\sum_{\langle \alpha, \beta \rangle}u_\alpha u_\beta\left(\frac{\psi_\alpha}{u_\alpha}-\frac{\psi_\beta}{u_\beta}\right)^2 + \sum_\alpha \frac{\psi_\alpha^2}{u_\alpha},\nonumber\\ &= E' \sum_\alpha \psi_\alpha^2,
\end{align}
where for the first term we use the well-known identity 
\begin{equation}
    \langle x| L |x\rangle = \sum_{\langle \alpha,\beta\rangle} w_{\alpha \beta} (x_\alpha - x_\beta)^2,
\end{equation}
which holds for any weighted Laplacian with $w_{\alpha \beta} = w_{\beta \alpha}$.  Equation \ref{eq:effschr} indicates that $1/u$ also acts as a confining potential.  The first term is the expected kinetic energy of the particle, which has been renormalized by the MBLL.  The second term denotes an expectation of a ``potential energy,'' if $1/u$ is treated as an effective potential.  However, this does not necessarily tell us any information of the behavior of eigenfunctions under this effective potential.  To analyze the decay of eigenfunctions, we will need to generalize  Agmon estimates to the Fock-space graph as discussed at the beginning of the section.

\subsection{Agmon Estimates for Eigenfunction Decay}

The intuition behind the derivation of Agmon estimates is to consider the behavior of an eigenstate modulated by an exponentially large amplitude and show that this new function is bounded in a certain sense.  This immediately shows that the eigenstate exponentially decays in particular regions.  As we will show, this exponential amplitude is obtained by minimizing an effective action on $\mathcal{G}_F$, and also depends on the connectivity of a particular node in the graph.

We first prove some identities obtained through brute-force computation.
\begin{lemma}
\label{Slemma}
Given a choice of $S_\alpha \in \mathbb{R}$, the following equation holds:
\begin{equation}
    \sum_\alpha \psi_\alpha^2 e^{2S_\alpha}\left[\frac{1}{u_\alpha}-E'-\frac{t}{2}\sum_{\beta \in N(\alpha)}(e^{S_\beta-S_\alpha} - 1)^2\right] \leq 0.
    \label{Eq:Slemma}
\end{equation}
\end{lemma}
\begin{proof}
We first compute the seemingly unrelated quantity $\mel{\frac{\psi}{u}e^{2S}}{M}{\psi}$.  From manipulating the first line of Equation \ref{eq:effschr}, this is equal to
\begin{align}
   \mel{\frac{\psi}{u}e^{2S}}{M}{\psi} = &t\sum_{\langle \alpha, \beta \rangle}u_\alpha u_\beta\left(\frac{\psi_\alpha}{u_\alpha} e^{2S_\alpha}-\frac{\psi_\beta}{u_\beta}e^{2S_\beta}\right)\times \nonumber \\ &\left(\frac{\psi_\alpha}{u_\alpha} - \frac{\psi_\beta}{u_\beta}\right) + \sum_\alpha \frac{\psi_\alpha^2}{u_\alpha} e^{2S_\alpha}.
\end{align}
Furthermore, manipulating the second line of Equation \ref{eq:effschr}, in particular utilizing the relationship between $M |\psi\rangle$ and the many-body energy $E$, the above quantity can be rewritten in terms of $E'$:
\begin{equation}
    \mel{\frac{\psi}{u}e^{2S}}{M}{\psi} = E' \sum_\alpha \psi_\alpha^2 e^{2S_\alpha}.
\end{equation}
By equating both of these equations and manipulating with simple algebra, we find the following identity:
\begin{align}
    t&\sum_{\langle \alpha, \beta \rangle}u_\alpha u_\beta\left(\frac{\psi_\alpha}{u_\alpha} e^{S_\alpha}-\frac{\psi_\beta}{u_\beta}e^{S_\beta}\right)^2 + \sum_\alpha \frac{\psi_\alpha^2}{u_\alpha} e^{2S_\alpha} - \nonumber \\ & t\sum_{\langle \alpha, \beta \rangle}\psi_\alpha \psi_\beta (e^{S_\alpha} - e^{S_\beta})^2 = E' \sum_\alpha \psi_\alpha^2 e^{2S_\alpha}.
\end{align}
Using the fact that $u_\alpha$ is positive, the first term is positive and the identity may be written as the inequality
\begin{equation}
    \sum_\alpha \psi_\alpha^2 e^{2S_\alpha}\left(\frac{1}{u_\alpha}-E'\right) - t\sum_{\langle\alpha, \beta\rangle}\psi_\alpha \psi_\beta (e^{S_\alpha} - e^{S_\beta})^2 \leq 0.
\end{equation}
Finally, we note the inequality $\psi_\alpha \psi_\beta \leq \frac{1}{2}(\psi_\alpha^2 + \psi_\beta^2)$ true for any $\psi$ and the fact that the second sum in the above equation is symmetric in the indices $\alpha$ and $\beta$.  Applying this inequality and symmetrizing the second sum with respect to $\alpha$ and $\beta$, we find Equation~(\ref{Eq:Slemma}),
where we add an additional factor of two from splitting the sum over edges to a double sum over nodes and avoiding double-counting pairs of nodes.
\end{proof}
Notice that the previous Lemma~\ref{Slemma} is general: it holds for an arbitrary choice of $S$.  Next, we choose $S$ so that the cumulative amplitude of the wavefunction in the classically forbidden region is ``small.''  By classically forbidden, we mean regions for which the many-body energy is smaller than the value of the effective potential $1/u$.  We make rigorous this notion in the following Lemma.
\begin{lemma}
There exists a choice $S_\alpha$ for all $\alpha$ such that for any $1>\epsilon > 0$, the energy inequality can be refined to
\begin{equation}
    \sum_\alpha \psi_\alpha^2 e^{2S_\alpha}\left(\frac{1}{u_\alpha}-E'\right) \leq (1-\epsilon) \sum_\alpha \psi_\alpha^2 e^{2S_\alpha}\left(\frac{1}{u_\alpha}-E'\right)_+,
\end{equation}
where $(x)_+ = \max(x,0)$.
\end{lemma}
\begin{proof}
First, define $\overline{U}$ to be the set of nodes $j$ which reside in the classically allowed region, where $\frac{1}{u_j}-E' < 0$.  Suppose that $S_\alpha$ satisfies:
\begin{equation}\label{eq:action}
    S_\alpha = \inf_{\substack{p(0) \in \overline{U} \\ p(-1) = \alpha}} \sum_{k} \left|\mathcal{L}_{p(k-1), p(k)}\right|,
\end{equation}
where the infimum is over paths from a point in a classically allowed region to $\alpha$.  By path, we mean a list of nodes such that each consecutive pair has an edge in $\mathcal{G}_F$; this is denoted by the vector $p$ where $p(i)$ is the $i$th element in the path.  The notation $p(-1)$ denotes the last element in the path.  We also assume that $|\mathcal{L}_{\alpha,\beta}| = |\mathcal{L}_{\beta,\alpha}|$.  It straightforwardly follows that, for $\beta \in N(\alpha)$:
\begin{align}
    S_\beta \leq S_\alpha + |\mathcal{L}_{\alpha,\beta}|,
\end{align}
and similarly $S_\beta \geq S_\alpha -|\mathcal{L}_{\beta,\alpha}|$.  Therefore, $|S_\beta - S_\alpha| \leq |\mathcal{L}_{\alpha,\beta}| = |\mathcal{L}_{\beta,\alpha}|$.  This essentially states that $S$ satisfies a triangle inequality, and therefore is a valid metric on $\mathcal{G}_F$.  Then, we may write
\begin{align}\label{eq:lagchoice}
    \frac{t}{2}\sum_{\beta \in N(\alpha)}(e^{S_\beta-S_\alpha} - 1)^2 &\leq \frac{t}{2}\sum_{\beta \in N(\alpha)}(e^{|S_\beta-S_\alpha|} - 1)^2,\nonumber\\
    &\leq \frac{t}{2}\sum_{\beta \in N(\alpha)}(e^{|\mathcal{L}_{\alpha,\beta}|} - 1)^2.
\end{align}
Next, we make the following choice for $|\mathcal{L}_{\alpha,\beta}| = |\mathcal{L}_{\beta,\alpha}|$: 
\begin{equation}
    |\mathcal{L}_{\alpha,\beta}| = \min_{k=\alpha,\beta}\left[\log\left(1+\sqrt{\frac{2(1-\epsilon)}{t |N(k)|}\left(\frac{1}{u_k}-E'\right)_+}\right)\right].
\end{equation}
With this choice, we may further bound Equation \ref{eq:lagchoice} by
\begin{align}
    \frac{t}{2}\!\!\sum_{\beta \in N(\alpha)}(e^{|\mathcal{L}_{\alpha,\beta}|} - 1)^2 &= (1-\epsilon)\!\!\!\sum_{\beta\in N(\alpha)}\min_{k=\alpha,\beta}\frac{\left({u_k}^{-1}-E'\right)_+}{|N(k)|}, \nonumber\\
    &\leq (1-\epsilon)\left(\frac{1}{u_\alpha}-E'\right)_+.
\end{align}
Substituting this inequality into the last term on the LHS of Equation \ref{Eq:Slemma}, we find the desired refinement of Equation \ref{Eq:Slemma}:
\begin{equation}
    \sum_\alpha \psi_\alpha^2 e^{2S_\alpha}\left(\frac{1}{u_\alpha}-E'\right) \leq (1-\epsilon) \sum_\alpha \psi_\alpha^2 e^{2S_\alpha}\left(\frac{1}{u_\alpha}-E'\right)_+.
\end{equation}
\end{proof}
The choice of variables $S$ and $\mathcal{L}$ suggests that the function $S_\alpha$ can be viewed as a saddle-point approximation on an effective action $S$ which is a sum of Lagrangian terms along a path from a node in a classically allowed region to $\alpha$.  This saddle-point value of the action will dictate the decay rate of the eigenfunction as defined on $\mathcal{G}_F$.  The statement that was just derived  roughly states that the sum of $\psi_\alpha^2 e^{2S_\alpha}$ over the classically allowed region $\overline{U}$ is greater than some constant times the sum of $\psi_\alpha^2 e^{2S_\alpha}$ over the classically disallowed region $U$.  Since $S_\alpha = 0$ for $\alpha \in \overline{U}$, the sum of the eigenfunctions in $U$ multiplied by a term exponential in the action $S$ is less than the sum of eigenfunctions in $\overline{U}$ without this exponential amplification.  Thus, in a sense, the value of an eigenfunction in $U$ must be exponentially small for its cumulative sum in $U$ to be small.  This intuition is made explicit in the following Theorem.

\begin{theorem}
Define $\overline{u} = \max_{\beta\in \overline{U}} u_\beta$.  For all nodes $\alpha$ in the classically forbidden region, an eigenstate with energy $E'$ satisfies
\begin{equation}
    |\psi_\alpha| \leq \sqrt{\frac{E' \overline{u}u_\alpha-u_\alpha}{\overline{u} - E' \overline{u}u_\alpha}} \left(\frac{e^{-S_\alpha}}{\sqrt{\epsilon}}\right),
\end{equation}
with $S_\alpha$ defined in Equation \ref{eq:action} and $1>\epsilon > 0$ a positive constant.
\end{theorem}
\begin{proof}
Based on the result in the previous Lemma, we may split the sum over nodes in $U$ and nodes in $\overline{U}$.  This gives
\begin{equation}
    \sum_{\alpha\in \overline{U}} \psi_\alpha^2 e^{2S_\alpha}\left(\frac{1}{u_\alpha}-E'\right) + \epsilon \sum_{\alpha\in U} \psi_\alpha^2 e^{2S_\alpha}\left(\frac{1}{u_\alpha}-E'\right) \leq 0.
\end{equation}
Since $S_\alpha = 0$ for $\alpha \in \overline{U}$, we obtain the following chain of inequalities:
\begin{align}\label{eq:ineqchain}
     \epsilon \sum_{\alpha\in U} \psi_\alpha^2 e^{2S_\alpha}\left(\frac{1}{u_\alpha}-E'\right) &\leq \sum_{\alpha\in \overline{U}} \psi_\alpha^2 \left(E'-\frac{1}{u_\alpha}\right),\nonumber\\
     &\leq \left(\sum_{\alpha\in \overline{U}} \psi_\alpha^2\right) \max_{\alpha \in \overline{U}} \left(E'-\frac{1}{u_\alpha}\right),\nonumber\\
     &\leq \max_{\alpha \in \overline{U}} \left(E'-\frac{1}{u_\alpha}\right).
\end{align}
Each term in the sum on the left hand side of the above equation must be strictly less than or equal to the right hand side.  Call $\overline{u} = \max_{\beta\in \overline{U}} u_\beta$.  This results in the bound
\begin{equation}
    |\psi_\alpha| \leq \sqrt{\frac{E'-\frac{1}{\overline{u}}}{\frac{1}{u_\alpha} - E'}} \left(\frac{e^{-S_\alpha}}{\sqrt{\epsilon}}\right) = \sqrt{\frac{E'\overline{u}u_\alpha-u_\alpha}{\overline{u} - E' \overline{u}u_\alpha}} \left(\frac{e^{-S_\alpha}}{\sqrt{\epsilon}}\right),
\end{equation}
where $\alpha \in U$.
\end{proof}
Therefore, the eigenfunction decays exponentially in the minimum value of the action.  The coefficient before the exponential has an intuitive meaning: it is the square root of the ratio of the height of the deepest classical well to the depth of penetration through the classically forbidden region, measured with respect to the many-body energy $E'$.  For instance, if the height to the deepest classical well is small, which occurs when the eigenstate is towards the bottom of the spectrum, then the coefficient scales as a constant, and decreases as the penetration depth into the classically forbidden region increases.  In the standard WKB formalism, one finds the following semi-classical decay for the wavefunction:
\begin{equation}
    \psi(x) \propto \frac{1}{\sqrt{p(x)}}  \text{exp}\left(-\int p(x)\,dx\right).
\end{equation}
Here, we find a similar expression, assuming that $p$ denotes a ``graph momentum,'' equal to $p_\alpha \sim \sqrt{(1/u_\alpha - E')/|N(\alpha)|}$.  The quantity within the exponential is a sum of Lagrangian terms $\mathcal{L} \sim \log(1 + p)$ along a path; this approximately reduces to a discrete WKB formula when the ``graph momentum'' is small.

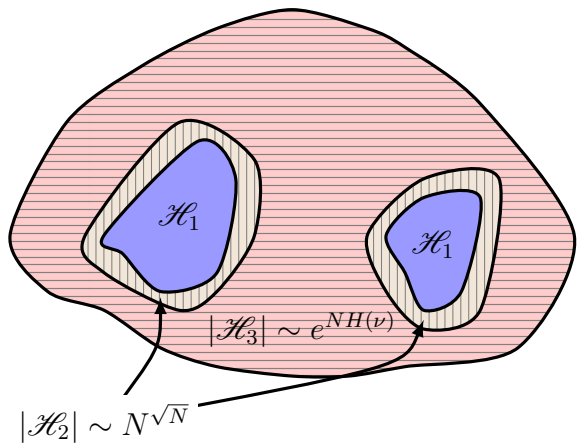
\begin{figure}
    \centering
    \begin{tikzpicture}[scale=2.5]
    \draw[very thick, pattern=horizontal lines, pattern color=black!50, preaction={fill=red!20}] plot[smooth cycle] coordinates {(-1.5,0) (-1.2,.6) (-.8,.9) (-.4,1.1) (0,1.2) (.4,1.1) (.6,1.0) (1,.7) (1.5,0) (1.2,-.6) (.6,-.7) (.3,-.75) (0,-.75) (-.4,-.7) (-.7,-.6) (-1,-.4) (-1.3,-.3)};
    \draw[very thick, pattern=vertical lines, pattern color=black!50, preaction={fill=brown!20}] plot[smooth cycle] coordinates {(-1.1,0) (-.6,.6) (-.2,.4) (-.2,0) (-.4,-.3) (-.6,-.4) (-1,-.2)};
    \draw[very thick, fill=blue!40] plot[smooth cycle] coordinates {(-1,0) (-.5,.5) (-.3,.3) (-.4,-.2) (-.7,-.3) (-.9,-.1)};
    
    \draw[very thick, pattern=vertical lines, pattern color=black!50, preaction={fill=brown!20}] plot[smooth cycle] coordinates {(.4,0) (.7,.3) (1.1,.3) (1,-.4) (.7,-.5) (.5,-.3)};
    \draw[very thick, fill=blue!40] plot[smooth cycle] coordinates {(.5,0) (.7,.2) (1,.2) (.9,-.3) (.7,-.4) (.6,-.2)};
    
    \node[font=\bfseries] at (0.05,-.5) {\large $\left|\mathscr{H}_3\right| \sim e^{N H(\nu)}$};
    \node at (.75,-0.05) {\large $\mathscr{H}_1$};
    \node at (-.6,0.1) {\large $\mathscr{H}_1$};
    \draw[very thick, ->,>=latex] (-1,-1) ..controls (-.7,-.6).. (-.7,-.33);
    \draw[very thick, ->,>=latex] (-1,-1) node[fill=white] {\large $\left|\mathscr{H}_2\right| \sim N^{\sqrt{N}}$} ..controls (.6,-.7).. (.7,-.45);
    
    \end{tikzpicture}
    \caption{The break-up of the full Hilbert space into the three regions described in the text.  One can visualize the Fock-space graph $\mathcal{G}_F$ as embedded into this Hilbert space.  The blue region (no texture) is the classically allowed area formed from a fixed many-body energy.  As shown, for sufficiently small energies, the classical regions will form islands in the full Hilbert space.  In the brown region (vertical lines), eigenstates decay, but at a rate does not increase with increasing system size.  The number of states in these regions is of order $N^{N^{1/2}}$.  In the red region (horizontal lines), the eigenstates decay at a rate scaling with an increasing function of system size; eigenstates decay exponentially for a constant fraction of the red region.}
    \label{fig:hilbspace}
\end{figure}

Next, we would like to understand the nature of the exponential decay in the action.  Each Lagrangian term scales like $\sim \sqrt{\left(1/u_\alpha-E'\right)/|N(\alpha)|}$, assuming that $|N(\alpha)|$ scales with system size in Equation \ref{eq:lagchoice}. Thus, the decay rate is crucially dictated by the size of $N(\alpha)$.  

For simplicity, let us assume $\mathcal{G}_P$ is a one-dimensional chain with nearest neighbor hopping.  Then, $|N(\alpha)|$ has the following behavior:
\begin{equation}\label{eq:degbound}
    \Theta(1) < |N(\alpha)| < \Theta(N).
\end{equation}
For states with a large domain of occupied sites, $|N(\alpha)|$ will scale like a constant.  In the worst case, if occupied and filled sites alternate, $|N(\alpha)|$ will scale linearly in $N$.  Thus, in the worst case, the action may only increase by $O(1/\sqrt{N})$ as one hops farther out in $\mathcal{G}_F$. 

Suppose we deal with these worst case states where $|N(\alpha)| \propto N$.  Further fix the many body energy $E'$, and denote $\mathscr{H}_1$ to the space of states living in the classically allowed regions.  As we hop along $\mathcal{G}_F$ and into the classically forbidden region the barrier height is at least a constant (we make an implicit assumption that we are considering typical paths in which the probability of encountering a node with barrier height scaling inversely with system size is small).  Suppose the minimizing path in Equation \ref{eq:action} to some node has length $n$ into the classically disallowed region; then, the action will scale like
\begin{equation}
    S \gtrsim \sum_{i=1}^n \log\left(1 + \sqrt{\frac{k_i}{N}}\right) \sim \frac{n}{N^{1/2}},
\end{equation}
where $k_i$ are unimportant constants.  Then, if $n = \Omega(N^{1/2})$, the action will begin to scale as a function of the system size.  Note that this bound is weak because the Lagrangian depends on $1/u-E'$, which increases along the minimizing path.  When this energy difference scales as $\Theta(N)$, we obtain exponential decay, which conceivably occurs for a fraction of the Hilbert space.  

With this intuition, we define three regions.  Each region covers nodes of the Fock space graph $\mathcal{G}_F$, but should be interpreted as the Hilbert space (i.e. the span of states associated with nodes in the region), denoted by $\mathscr{H}$:
\begin{itemize}
    \item $\mathscr{H}_1$: Fixing a value of the many-body energy $E'$, this region covers classically allowed regions of $\mathcal{G}_F$ where the wavefunction behaves like a constant.
    
    \item $\mathscr{H}_2$: This region is an annulus of thickness $N^{1/2}$ around $\mathscr{H}_1$ regions.  It corresponds to the intermediate region in which there exist portions that do not decay at a rate scaling with the system size.  In the thermodynamic limit, we will not observe decay in these portions.  However, there may exist several portions of $\mathscr{H}_2$ where we do find decay scaling with system size.
    
    \item $\mathscr{H}_3$: This region is outside of $\mathscr{H}_1$ and $\mathscr{H}_2$, in which eigenstates almost surely decay at a rate that scales with some increasing function in the system size.  For instance, if $n = \Theta(N^{1/2}\log N)$, then the eigenstates obey a power-law decay.  Eigenstates will decay exponentially when $1/u-E' = \Theta(N)$.
\end{itemize}


These regions are illustrated in Figure \ref{fig:hilbspace}.  In addition, since the degree of each node is at most $N$, the dimension $\mathscr{H}_2$ is roughly $N^{\sqrt{N}}$, which is of zero-measure with respect to the dimension of the full Hilbert space, roughly equal to $e^{N H(\nu)}$ (here, we make the assumption that $N$ is a fraction of $L$, so that $\binom{L}{N} \sim e^{N H(\nu)}$ where $\nu$ is the filling fraction and $H(\nu) = -\nu \log \nu - (1-\nu)\log(1-\nu)$).  For $\mathscr{H}_1$ regions centered about states with degree that is $o(N)$, the estimates derived above can be refined by substituting $N$ with the degree.  Thus, based on the bound in Equation \ref{eq:degbound}, there can exist some regions of the full Hilbert space that exhibit strong exponential localization.

It is worth noting what changes as the underlying dimension increases.  For instance, if $\mathcal{G}_P$ is a $d$-dimensional lattice, then the topology of $\mathcal{G}_F$ becomes far more connected.  In particular, the degree distribution follows
\begin{equation}
    \Theta\left(N^{1-\frac{1}{d}}\right) < |N(\alpha)| < \Theta(N),
\end{equation}
with minimum degree nodes having the property that linearly-sized domains are occupied, and maximum degree nodes having a large number of constant-size domains.  Note that the minimum degree corresponds to an area law, while the maximum degree corresponds to a volume law.  The bounds on decay rate developed in the preceding paragraph follow immediately by substituting $N$ with the degree. The higher connectivity present in these graphs could point to the existence of a mobility edge in higher dimensions at least in the Anderson model, as will be further elaborated in Section V.

\subsection{Predicting Regions of Localization}

In the previous section, we have addressed the question of what occurs in classically forbidden regions. Now, we explore how the eigenstates behave in the classically-allowed regions.  In the single-particle case, it was argued that so long as these regions are far apart and not resonant, the eigenstates (towards the bottom of the spectrum) tend to localize at only one of these regions, and have exponentially small amplitudes on the others \cite{potential-2}.  Having a similar result in the many-body case would support any assertion of MBL.

To proceed, it is first important for us to precisely define the regions of interest.  As before, denote $\overline{U}$ the set of classically allowed regions which satisfy $1/u_\alpha - E' \leq 0$.  These regions/wells form sets of disjoint connected regions in $\mathcal{G}_F$, which we label by the index $\ell$.  Let us fix a particular such well $\overline{U}(\ell)$; compute
\begin{equation}
    \overline{S} = \inf \{S(\alpha, \beta): \alpha \in \overline{U}(\ell), \beta \in \overline{U}(\ell')\},
\end{equation}
for $\ell' \neq \ell$.  The notation $S(\alpha, \beta)$ denotes the action of a path from node $\alpha$ to $\beta$.  Essentially, this quantity measures the minimum value of the action along any path connecting well $\ell$ to any other well.  Call $W(\ell)$ the set of nodes $\alpha$ such that $\inf \{S(\alpha, \beta): \beta \in U(\ell)\} \leq \overline{S}/2$.  This defines a neighborhood around well $\ell$ consisting of points for which $\ell$ is the closest such well.  Repeating for all other wells, it is clear that $\bigcup_\ell W(\ell)$ is the set of all nodes in $\mathcal{G}_F$ and $W(\ell)$ is disjoint from $W(\ell')$ (we may arbitrarily assign tie-breaker nodes shared by $W(\ell)$ and $W(\ell')$ at the boundary to assure disjointness). 

In this section, we will substantially make use of the test state $\ket{\zeta} = \sum_{a \in A} \gamma_a \ket{\psi^{(a)}}$, which is written in the eigenbasis (and we enumerate eigenstates using the Latin alphabet), where $A$ is specified in the theorems stated below.  Further define $\ket{\phi^{(\ell)}}$ to be the eigenstates of the eigenvalue problem restricted to nodes in $W(\ell)$ (with zero boundary conditions on all other nodes).  For the remainder of the theorem below, we drop the label $\ell$ in $\ket{\phi^{(\ell)}}$.  
We will first need the following Lemma in order to proceed:
\begin{lemma}
Consider the subgraph induced by the nodes in $W(\ell)$, which we call $\mathcal{G}_W$, and solve the eigenvalue problem on $\mathcal{G}_W$ with zero boundary conditions applied to all nodes not in $W(\ell)$.  Call an eigenstate of this problem $\ket{\psi^{(W)}}$ for which the energy of such a state is $E'_W$ upper bounded by $E'$.  Call $u$ the landscape for the original eigenvalue problem on the entirety of $\mathcal{G}_F$, and $u^{(W)}$ the landscape for the eigenvalue problem on the subgraph induced by $W(\ell)$.  Then, 
\begin{equation}
    |\psi^{(W)}_\alpha| \leq \sqrt{\frac{E_W'-\frac{1}{\overline{u}^{(W)}}}{\frac{1}{u^{(W)}_{\alpha}} - E_W'}} \left(\frac{e^{-S_\alpha}}{\sqrt{\epsilon}}\right) \leq O\left(\sqrt{\frac{N}{\sigma}}\right) e^{-S_\alpha},
\end{equation}
for some $\epsilon > 0$ and $S_\alpha$ the action of the eigenvalue problem on all of $\mathcal{G}_F$ with energy $E'$.  The notation $\overline{u}^{(W)}$ indicates the maximum value of $u^{(W)}$ in classically allowed regions of $W(\ell)$ where $1/u^{(W)}-E_W' < 0$.  Furthermore, the second inequality follows from $E_W'-\frac{1}{\overline{u}^{(W)}} = O(N)$ and $\frac{1}{u^{(W)}_{\alpha}} - E_W' > \sigma$ in classically disallowed regions for some $\sigma > 0$.  
\end{lemma}
\begin{proof}
To prove this statement, it suffices to show that $1/u^{(W)} - E_W'$ for the subgraph problem is greater than $1/u - E'$ for the entire problem; therefore, the action for the subgraph problem will be larger than that for the entire problem.  First, we know that
\begin{equation}
    -t \sum_{\beta \in N(\alpha)} u_{\beta} + \left(\sum_{i = 1}^{L} V n_i^{(\alpha)} n_{i+1}^{(\alpha)} + \epsilon_i n_i^{(\alpha)}+K \right)u_\alpha = 1
\end{equation}
for the landscape on the entire domain, and 
\begin{equation}
    -t \sum_{\beta \in N_W(\alpha)} u^{(W)}_{\beta} + \left(\sum_{i = 1}^{L} V n_i^{(\alpha)} n_{i+1}^{(\alpha)} + \epsilon_i n_i^{(\alpha)}+K \right)u^{(W)}_\alpha = 1,
\end{equation}
where $N_W(\alpha)$ denotes the set of neighbors of $\alpha$ in the subgraph induced by the well nodes.  Note that $N_W(\alpha)$ is a subset of $N(\alpha)$.  Subtracting the two equations above for all $\alpha \in W(\ell)$ gives
\begin{align}
    -t \sum_{\beta \in N_W(\alpha)} &\left(u_{\beta}-u^{(W)}_{\beta}\right) + \left(\sum_{i = 1}^{L} V n_i^{(\alpha)} n_{i+1}^{(\alpha)} + \epsilon_i n_i^{(\alpha)}+K \right)\nonumber \\ &\times \left(u_{\alpha}-u^{(W)}_{\alpha}\right) = t \sum_{\beta \in N(\alpha)\setminus N_W(\alpha)} u_{\beta}.
\end{align}
From Lemma III.1, the RHS of this equation is positive, and by the same Lemma, this implies that $\left(u_{\alpha}-u^{(W)}_{\alpha}\right) \geq 0$.  Therefore, $1/u_{\alpha} \leq 1/u^{(W)}_{\alpha}$ and by the assumption that $E'_W \leq E'$, we have that $1/u_W - E_W' \geq 1/u - E'$.  Thus, utilizing the bounds presented in the previous subsection along with this result, we prove this Lemma.
\end{proof}

We need to make note of some additional notation.  We call $\mathcal{G}_W$ the subgraph induced by nodes in the well, and we call $\partial \mathcal{G}_W$ the set of edges connecting nodes in $W(\ell)$ to nodes outside of $W(\ell)$.  With this, we may now state the following locality theorem:
\begin{theorem}
{\bf (Locality theorem I):} Suppose $\ket{\phi}$ is an eigenstate with eigenvalue $\mu$ of the eigenvalue problem on the subgraph induced by nodes in $W(\ell)$.  Denote by $P_\psi(\mu-\delta, \mu+\delta)$ the projection operator onto the basis of eigenstates (in the full graph) with eigenvalues between $\mu-\delta$ and $\mu+\delta$; that is
\begin{equation}
    P_\psi(\mu-\delta, \mu+\delta) = \sum_{a \in \overline{A}} \ket{\psi^{(a)}}\bra{\psi^{(a)}},
\end{equation}
where $A$ is the set of eigenstates with eigenvalues $\lambda_a$ satisfying $\abs{\lambda_a - \mu} \geq \delta$.  Then,
\begin{equation}
\norm{(I-P_\psi(\mu-\delta, \mu+\delta))\ket{\phi}}_2 \leq O\left(\sqrt{\frac{N}{\sigma}}\right) \frac{N}{\delta} e^{-\overline{S}/2},
\end{equation}
where $\sigma$ is defined in the previous Lemma.
\end{theorem}
\begin{proof}
Throughout this proof, we will be using labels such as $\ket{fg}$ to mean $\sum_\alpha f_\alpha g_\alpha \ket{\alpha}$.  These do not represent physical quantum states because they are not normalized.  

First recall that operators such as $M$ defined at the beginning of this section operate on the entirety of $\mathcal{G}_F$.  For this proof, we will assume that they only operate on $\mathcal{G}_W$.  Since $\ket{\phi}$ has eigenvalue $\mu$, we know that $M \ket{\phi} = \mu \ket{u^{(W)} \phi}$.  Let us instead consider the state $\ket{\eta \phi}$, for some $\eta$ to be presented momentarily.  Then, $M \ket{\eta \phi} = \mu \ket{u \eta \phi} + \ket{u r}$ for some residual state $\ket{r}$.  We note that $u$ here is the landscape defined for the global $\mathcal{G}_F$.  Applying the test state $\bra{\zeta/u}$ to this equation gives us the following expression for $\bra{\zeta/u}\ket{u r} = \bra{\zeta}\ket{r}$:
\begin{align}\label{eq:zeta}
    \bra{\zeta}\ket{r} = t\sum_{\langle \alpha, \beta \rangle \in \mathcal{G}_W}u_\alpha u_\beta&\left(\frac{\zeta_\alpha}{u_\alpha} -\frac{\zeta_\beta}{u_\beta}\right)\left(\frac{\phi_\alpha \eta_\alpha}{u_\alpha} - \frac{\phi_\beta \eta_\beta}{u_\beta}\right) \nonumber\\ &+ \sum_\alpha \zeta_\alpha \phi_\alpha \eta_\alpha\left(\frac{1}{u_\alpha} - \mu\right).
\end{align}
This expression may be written as
\begin{align}
    \bra{\zeta}\ket{r} =  t&\sum_{\langle \alpha, \beta \rangle \in \mathcal{G}_W \cup \partial \mathcal{G}_W}u_\alpha u_\beta\left(\frac{\zeta_\alpha}{u_\alpha} -\frac{\zeta_\beta}{u_\beta}\right)\left(\frac{\phi_\alpha \eta_\alpha}{u_\alpha} - \frac{\phi_\beta \eta_\beta}{u_\beta}\right) \nonumber\\ &- t\sum_{\langle \alpha, \beta \rangle \in \partial \mathcal{G}_W}u_\alpha u_\beta \left(\frac{\zeta_\alpha}{u_\alpha} -\frac{\zeta_\beta}{u_\beta}\right)\left(\frac{\phi_\alpha \eta_\alpha}{u_\alpha}- \frac{\phi_\beta \eta_\beta}{u_\beta}\right) \nonumber \\&+ \sum_\alpha \zeta_\alpha \phi_\alpha \eta_\alpha\left(\frac{1}{u_\alpha} - \mu\right).
\end{align}
Expanding $\ket{\zeta}$ via its definition into the basis of eigenstates in the set $A$ results in
\begin{alignat}{2}
    \bra{\zeta}\ket{r} =  t&\sum_{a \in A} \gamma_{a}\sum_{\langle \alpha, \beta \rangle \in \mathcal{G}_W \cup \partial \mathcal{G}_W}&&u_\alpha u_\beta\left(\frac{\psi^{(a)}_\alpha}{u_\alpha} -\frac{\psi^{(a)}_\beta}{u_\beta}\right)\nonumber\\& &&\times\left(\frac{\phi_\alpha \eta_\alpha}{u_\alpha} - \frac{\phi_\beta \eta_\beta}{u_\beta}\right) \nonumber\\ &+\sum_{a \in A} \gamma_{a} \sum_{\alpha \in W} \psi^{(a)}_\alpha \phi_\alpha \eta_\alpha&&\left(\frac{1}{u_\alpha}-\lambda_{a}\right) \nonumber\\ &- t\sum_{a \in A}\gamma_{a}\sum_{\langle \alpha, \beta \rangle \in \partial \mathcal{G}_W}&&u_\alpha u_\beta \left(\frac{\psi^{(a)}_\alpha}{u_\alpha} -\frac{\psi^{(a)}_\beta}{u_\beta}\right)\nonumber \\ & &&\times\left(\frac{\phi_\alpha \eta_\alpha}{u_\alpha}- \frac{\phi_\beta \eta_\beta}{u_\beta}\right) \nonumber \\&+ \sum_{a \in A} \gamma_{a} \sum_{\alpha \in W} \psi^{(a)}_\alpha \phi_\alpha \eta_\alpha &&\left(\lambda_{a} - \mu\right).
\end{alignat}
The first two terms can be combined to equal zero from the eigenvalue equation $M \ket{\psi_\alpha} = \lambda_\alpha \ket{u \psi_\alpha}$.  The addition of the edges $\partial \mathcal{G}_W$ in the summation was necessary so that the boundary nodes satisfy the appropriate eigenvalue equation for the entire domain.  The residual term (the third term) can be made to be zero by forcing $\eta = 0$ on all nodes in $\partial \mathcal{G}_W$.  Thus, with this choice of $\eta$, we are left with the fourth term:
\begin{equation}\label{eq:eigbasis}
    \bra{\zeta}\ket{r} = \sum_{a\in A} \gamma_a \sum_\alpha \psi^{(a)}_\alpha \phi_\alpha \eta_\alpha\left(\lambda_a - \mu\right),
\end{equation}
where the remaining terms in Equation \ref{eq:zeta} disappear from application of the eigenvalue equation for $\ket{\psi^{(a)}}$.  Denoting the inner product $\bra{\psi^{(a)}}\ket{\eta \phi} = \nu_a$, we find
\begin{equation}\label{eq:firstbound}
    \bra{\zeta}\ket{r} = \sum_{a\in A} \gamma_a \nu_a \left(\lambda_a - \mu\right).
\end{equation}
Next, we derive a bound on $\bra{\zeta}\ket{r}$.  We start by noting that the following identity is true:
\begin{align}\label{eq:eigidentity}
    \mel{\frac{\eta \zeta}{u}}{M}{\phi}& - \mu \bra{\eta \zeta}\ket{\phi} = t\sum_{\langle \alpha, \beta \rangle \in \mathcal{G}_W}u_\alpha u_\beta\left(\frac{\zeta_\alpha\eta_\alpha}{u_\alpha} -\frac{\zeta_\beta\eta_\beta}{u_\beta}\right)\nonumber\\ &\times \left(\frac{\phi_\alpha}{u_\alpha} - \frac{\phi_\beta}{u_\beta}\right) + \sum_{\alpha \in W} \zeta_\alpha \phi_\alpha \eta_\alpha\left(\frac{1}{u_\alpha} - \mu\right).
\end{align}
This expression in fact equals zero.  To show this, let us rewrite the eigenvalue problem for $\ket{\phi}$ in terms of the \emph{global} landscape $\ket{u}$.  Following the analysis of Section IV.A and replacing the local landscape by the global landscape, we find that Equation \ref{eq:landschromassaged} becomes
\begin{equation}
- t \sum_{\beta \in N(\alpha)} u_\beta u_\alpha \left(\frac{\phi_\beta}{u_\beta}-\frac{\phi_\alpha}{u_\alpha}\right) + \phi_\alpha = \mu u_\alpha \phi_\alpha,
\end{equation}
so long as $\alpha$ is located in the bulk of $W$ and not at the boundary $\partial W$.  Then, since $\eta_\alpha = 0$ at the boundary, the quantity $\mel{\frac{\eta \zeta}{u}}{M}{\phi} - \mu \bra{\eta \zeta}\ket{\phi}$ does not depend on the boundary behavior.  In the bulk, $M \ket{\phi} = \mu \ket{u \phi}$, so the quantity equals zero, as desired.

Subtracting Equation \ref{eq:eigidentity} from Equation \ref{eq:zeta} and simple algebraic manipulation gives the following expression:
\begin{equation}
    \bra{\zeta}\ket{r} = t\sum_{\langle \alpha, \beta \rangle} (\eta_\beta - \eta_\alpha)(\phi_\alpha \zeta_\beta - \phi_\beta \zeta_\alpha).
\end{equation}
Next, we choose $\eta_\alpha$ to equal $1$ if $\inf \{S(\alpha, \beta): \beta \in U(\ell)\} \leq \overline{S}/2 - \delta'$ and zero otherwise.  We also have the additional restriction that $\eta_\alpha = 0$ for nodes at the boundary $\partial W$. Therefore, the sum above has nonzero contribution only if the edge $\langle \alpha, \beta \rangle$ connects a node where $\eta = 0$ and a node where $\eta = 1$.  For $\delta'$ chosen small enough, the sum will only have contribution in the subgraph $\partial \mathcal{G}_W$.  We may then write
\begin{equation}
    \bra{\zeta}\ket{r} = t \mel{\zeta}{P}{\phi},
\end{equation}
where $P$ denotes a matrix with elements $P_{\alpha \beta} = \eta_\beta - \eta_\alpha$.  Next, we perform a restriction of both $\ket{\phi}$ and $\ket{\zeta}$ as well as $P$ to nodes $\alpha$ for which $P_{\alpha \beta} \neq 0$ for some $\beta$.  By the Cauchy-Schwartz inequality, we find
\begin{equation}\label{eq:cauchy}
    \abs{\bra{\zeta}\ket{r}}^2 \leq t^2 \sigma^2_{\text{max}}(P) \bra{\phi}\ket{\phi} \bra{\zeta}\ket{\zeta},
\end{equation}
where $\sigma_{\text{max}}(P)$ denotes the maximum singular value of $P$, or alternately the 2-norm.  First, we have $\bra{\zeta}\ket{\zeta} \leq \sum_{a \in A} \gamma_a^2$.  Next, we refer to Lemma IV.4.  Since $\ket{\phi}$ is an eigenstate on $\mathcal{G}_W$, then from the Lemma and Equation \ref{eq:ineqchain} 
\begin{equation}
    \epsilon \sum_{\alpha \in U} \phi_\alpha^2 e^{2S_\alpha} \leq \frac{\max_{\alpha \in \overline{U}} \Delta^{(W)}_\alpha}{\min_{\alpha \in U} -\Delta^{(W)}_\alpha} \leq O\left(\frac{N}{\sigma}\right),
\end{equation}
where $\Delta_\alpha = 1/u_\alpha - E'$ and $\Delta^{(W)}_\alpha = 1/u^{(W)}_\alpha - E_W'$.  Furthermore, $U$ and $\overline{U}$ denote classically disallowed and classically allowed regions respectively.  Since the action of all nodes $\alpha \in \partial W$ is at least $\overline{S}/2-\delta'$, then we have (assuming $\delta'$ is constant and can be ignored)
\begin{equation}
    \sum_{\alpha \in \partial W} \phi_\alpha^2 \leq \frac{e^{-\overline{S}}}{\epsilon}\frac{\max_{\alpha \in \overline{U}} \Delta^{(W)}_\alpha}{\min_{\alpha \in U} -\Delta^{(W)}_\alpha} \leq O\left(\frac{N}{\sigma}\right) e^{-\overline{S}}.
\end{equation}
Next, we note the well-known identity (which is a special case of Holder's inequality)
\begin{equation}
    \sigma_{\text{max}}(P) \leq \sqrt{\norm{P}_1 \norm{P}_\infty} \leq kN,
\end{equation}
where the second inequality follows from the fact that the 1-norm and the $\infty$-norm are upper bounded by the degree of any node in $\mathcal{G}_F$, and $k$ is some constant.  Recall that the 1-norm of matrix $P$ is the maximum column sum of $\abs{P}$, while the $\infty$-norm is the maximum row sum of $\abs{P}$.  Utilizing all of the bounds derived, Equation \ref{eq:cauchy} becomes
\begin{align}
    \abs{\bra{\zeta}\ket{r}}^2 &\leq \frac{k^2 t^2 N^2}{\epsilon} \frac{\max_{\alpha \in \overline{U}} \Delta^{(W)}_\alpha}{\min_{\alpha \in U} -\Delta^{(W)}_\alpha} e^{-\overline{S}} \left(\sum_{a \in A} \gamma_a^2\right),\nonumber\\ &\leq O\left(\frac{N^3}{\sigma}\right) e^{-\overline{S}} .
\end{align}
Next, return to Equation \ref{eq:firstbound}, and choose $\gamma_a = \nu_a \text{sgn}(\lambda_a - \mu)$ for all $a \in A$.  Then, combining Equation \ref{eq:firstbound} with the above, we find that
\begin{equation}
   \left(\sum_{a\in A} \nu^2_a \abs{\lambda_a - \mu}\right)^2 \leq O\left(\frac{N^{3}}{\sigma}\right) e^{-\overline{S}} \left(\sum_{a \in A} \nu_a^2\right),
\end{equation}
where we have suppressed the unnecessary constants for convenience.  Because $|\lambda_a - \mu| \geq \delta$, we obtain the inequality
\begin{equation}
    \sum_{a \in A} \nu_a^2 \leq \frac{O(N^3)}{\delta^2 \sigma} e^{-\overline{S}}.
\end{equation}
However, the LHS of this expression is just equal to the norm squared of the orthogonal complement of $\ket{\eta \phi}$ onto the span of the eigenvectors in $A$.  In terms of notation described in the theorem statement, this may be written as
\begin{equation}\label{eq:etaphi}
    \norm{(I-P_\psi(\mu-\delta, \mu+\delta))\ket{\eta \phi}}_2 \leq O\left(\sqrt{\frac{N}{\sigma}}\right) \frac{N}{\delta} e^{-\overline{S}/2}.
\end{equation}
Next, we note that $\bra{(1-\eta)\phi}\ket{(1-\eta)\phi}$ satisfies the following bound:
\begin{equation}
    \bra{(1-\eta)\phi}\ket{(1-\eta)\phi} \leq \frac{\max_{\alpha \in \overline{U}} \Delta_\alpha}{\min_{\alpha \in U} -\Delta_\alpha} e^{-\overline{S}} \leq O\left(\frac{N}{\sigma}\right) e^{-\overline{S}},
\end{equation}
which follows both from the fact that $\ket{(1-\eta)\phi}$ has support on nodes with action greater than $\overline{S}/2-\delta'$ as well as Equation \ref{eq:ineqchain} and Lemma IV.4.  Because $I-P_\psi(\mu-\delta, \mu+\delta)$ has matrix norm equal to one, we have
\begin{equation}
    \norm{(I-P_\psi(\mu-\delta, \mu+\delta))\ket{(1-\eta) \phi}}_2 \leq O\left(\sqrt{\frac{N}{\sigma}}\right) e^{-\overline{S}/2}.
\end{equation}
Adding this to Equation \ref{eq:etaphi}, and using the triangle inequality, we find
\begin{equation}
    \norm{(I-P_\psi(\mu-\delta, \mu+\delta))\ket{\phi}}_2 \leq O\left(\sqrt{\frac{N}{\sigma}}\right) \frac{N}{\delta} e^{-\overline{S}/2},
\end{equation}
where again we have absorbed all constants or lower-order contributions into the $O\left(\sqrt{\frac{N}{\sigma}}\right)$ term.
\end{proof}
In short, this theorem states that solving the eigenvalue problem in the well $W(\ell)$ gives an eigenstate that is very close to the span of eigenstates in the entire domain with energies sufficiently close.  To show that the well eigenstates are themselves similar to the full eigenstates requires us to prove an additional statement.  Before presenting the theorem, we note the change in notation and denote $\ket{\phi^{(\ell, \alpha)}}$ to be the $\alpha$-th eigenstate of the eigenvalue problem in well $W(\ell)$.
\begin{theorem}
{\bf (Locality theorem II): }Suppose $\ket{\psi}$ is an eigenstate with eigenvalue $\mu$ of the eigenvalue problem on the entire graph $\mathcal{G}_F$.  Denote by $P_\phi(\mu-\delta, \mu+\delta)$ the projection operator onto the basis of the union of all eigenstates in each well $W(\ell)$ with eigenvalues between $\mu-\delta$ and $\mu+\delta$; that is
\begin{equation}
    P_\phi(\mu-\delta, \mu+\delta) = \sum_{(\ell, a) \in \overline{A'}} \ket{\phi^{(\ell,a)}}\bra{\phi^{(\ell,a)}},
\end{equation}
where $A'$ is the set of pairs $(\ell,a)$ of eigenstates and corresponding wells with eigenvalues $\lambda_{\ell,a}$ satisfying $\abs{\lambda_{\ell,a} - \mu} \geq \delta$.  Then,
\begin{equation}
\norm{(I-P_\phi(\mu-\delta, \mu+\delta))\ket{\psi}}_2 \leq O\left(\sqrt{\frac{N}{\sigma}}\right) \frac{N}{\delta} e^{-\overline{S}/2}.
\end{equation}
\end{theorem}
\begin{proof}
The proof is very similar to that of the previous Theorem, but there are several key changes.  Most of the changes involve swapping $\phi$ and $\psi$, and thus we will be very brief with the proof.  First, define $\ket{\psi}$ and note that it satisfies the eigenvalue equation.  Consider the restriction of $\ket{\psi}$, given by $\ket{\eta \psi}$, which satisfies $M \ket{\psi} = \mu \ket{u \eta \psi} + \ket{u r}$, again for some residual vector $\ket{r}$.

Next, define the test state $\ket{\zeta} = \sum_{(\ell,a) \in A'} \gamma_{\ell, a} \ket{\phi^{(\ell,a)}}$; the difference here is that we take a superposition of different eigenstates from different wells if their corresponding eigenvalue is close to $\mu$.  Then, Equation \ref{eq:eigbasis} becomes
\begin{equation}
    \bra{\zeta}\ket{r} = \sum_{(\ell,a)\in A'} \gamma_{\ell,a} \sum_\alpha \phi^{(\ell, a)}_\alpha \psi_\alpha \eta_\alpha\left(\lambda_{\ell,a} - \mu\right),
\end{equation}
Before, to arrive at this equation, we needed to set $\eta$ at the boundary nodes equal to zero; here, a similar manipulation holds except $\eta$ will equal zero on nodes at the boundary of each well $\ell$ appearing in $A'$.  Moreover, $u$ is the landscape for the global problem, not the well problem for $\ket{\phi^{(\ell,a)}}$, but this can be amended by using the fact that $\ket{\phi}$ solves the eigenvalue problem with the global landscape for nodes in the bulk of the well as well as $\eta = 0$ in the boundary of the wells.  This is identical to the reasoning used for setting Equation \ref{eq:eigidentity} to zero.  

With this, Equation \ref{eq:firstbound} then becomes
\begin{equation}
    \bra{\zeta}\ket{r} = \sum_{(\ell,a)\in A'} \gamma_{\ell,a} \nu_{\ell,a}\left(\lambda_{\ell,a} - \mu\right),
\end{equation}
where $\nu_{\ell,a} = \bra{\phi^{(\ell,a)}}\ket{\eta \psi}$.  After some basic manipulation as in the previous theorem, we have the following expression:
\begin{equation}
    \bra{\zeta}\ket{r} = t\sum_{\langle \alpha, \beta \rangle} (\eta_\beta - \eta_\alpha)(\psi_\alpha \zeta_\beta - \psi_\beta \zeta_\alpha).
\end{equation}
Next, since this sum is over the entire Fock space graph, we select $\eta_\alpha = 1$ if $S_\alpha \leq \overline{S}/2-\delta'$ and $\eta_\alpha = 0$ otherwise.  Roughly speaking, we draw thick boundaries around each of the wells, and it is within these regions where $\eta = 0$.  As before, for sufficiently small $\delta'$, the region in which $\eta = 0$ will include all nodes that we require $\eta = 0$ at.  As a result, the sum above is only nonzero for special edges $\langle \alpha, \beta \rangle \in \bigcup_\ell \partial \mathcal{G}_W(\ell)$ where $\eta_\alpha = 0$ and $\eta_\beta = 1$ or vice versa; furthermore, the nodes for which $\eta$ is nonzero have action at least equal to $\overline{S}/2-\delta'$.  The analysis then proceeds identically to the previous theorem (with $\phi$ replaced with $\psi$), and use of Cauchy-Schwartz along with the other identities gives
\begin{align}
    \abs{\bra{\zeta}\ket{r}}^2 &\leq \frac{k^2 t^2 N^2}{\epsilon} \frac{\max_{\alpha \in \overline{U}} \Delta_\alpha}{\min_{\alpha \in U} -\Delta_\alpha} e^{-\overline{S}} \left(\sum_{(\ell,a) \in A'} \gamma_{\ell,a}^2\right),\nonumber\\ &\leq O(N^3) e^{-\overline{S}}.
\end{align}
In this case, the analysis is simpler because we do not need to make use of Lemma IV.4.  Next, we choose $\gamma_{\ell,a} = \nu_{\ell,a} \text{sgn}\left(\lambda_{\ell,a}-\mu\right)$, which leads to the equation
\begin{equation}
    \sum_{(\ell,a) \in A'} \nu_{\ell,a}^2 \leq \frac{O(N^3)}{\delta \sigma} e^{-\overline{S}/2},
\end{equation}
leading to 
\begin{equation}
    \norm{(I-P_\phi(\mu-\delta, \mu+\delta))\ket{\eta \psi}}_2 \leq O\left(\sqrt{\frac{N}{\sigma}}\right) \frac{N}{\delta} e^{-\overline{S}/2}.
\end{equation}
Finally, we have a similar bound for $\ket{(1-\eta)\psi}$ by utilizing the same logic as in the last theorem.  Thus, we are done.
\end{proof}
This statement shows that a global eigenstate $\ket{\psi}$ is close to the span of well eigenstates with similar eigenvalues.  Note that this (coupled with the last theorem) presents a bijective correspondence between well and global eigenstates.  In particular, if the spectrum associated with a well is not resonant with that of any another well, then, any global eigenvalue that is close enough to an eigenvalue of the first well will imply that the global eigenstate will localize in that well.  We call these ``locality theorems'' as they argue that eigenstates do not uniformly spread over all possible classically allowed regions, but rather tend to only localize at a small handful.  Furthermore, both these theorems have useful bounds if the minimum Agmon distance between wells is of order $\log(N)$, otherwise the RHS of both theorems will scale directly with system size.

One can convert both of these theorems into statements comparing the global spectrum with the well spectrum.  Specifically, call $C(\lambda) = \#\{\lambda_\alpha : \lambda_\alpha < \lambda\}$ and $C_0(\lambda) = \#\{\lambda_{\ell, \alpha} : \lambda_{\ell, \alpha} < \lambda\}$, which counts the number of eigenvalues in the global and local problems, respectively, below a particular value $\lambda$.  From an identical analysis as in Ref. \cite{potential-2}, we then arrive at the following Corollary:
\begin{corollary}
Suppose we choose $\overline{C}$ and $\delta$ such that
\begin{equation}
    O\left(\sqrt{\frac{N}{\sigma}}\right) \frac{N \overline{C}}{\delta} < e^{\overline{S}/2}.
\end{equation}
Then, it follows that
\begin{align}
    \min(\overline{C}, C_0(\mu - \delta)) &\leq C(\mu), \nonumber\\
    \min(\overline{C}, C(\mu - \delta)) &\leq C_0(\mu).
\end{align}
\end{corollary}
\begin{proof}
See Ref. \cite{potential-2}, Corollary 5.2.
\end{proof}

This corollary essentially states that the global and the well spectrum can be $\delta$-close to each other up to an upper counting limit of $\overline{C}(\delta)$.  This gives the intuition that the eigenvalues of the global and well problems more closely agree towards the bottom of the spectrum (since $\delta$ can be chosen to be smaller without significantly decreasing the size of $\overline{C}$) -- this indicates that global eigenstates are likely to only resonate with a single well towards the bottom of the spectrum.

\section{Locator Expansion for MBLL}

In the previous two sections, we have defined a many-body generalization of the LL, $u$, and showed that it satisfies the same identities as derived by Mayboroda et al.~\cite{pnas,ll-dual} in the single particle case.  We further showed that $1/u$ behaves as an effective potential and derived decay rates that relate to the height of an effective potential barrier and the connectivity of a particular node in the Fock space graph.  However, the effective potential is in general an analytically intractable object to calculate, because it requires one to invert the Hamiltonian.  We would thus like to calculate corrections to the effective potential perturbatively, much like Anderson has done using the famous Locator Expansion (LE)~\cite{Anderson}.  In this section, we show that naive perturbation theory does not suffer from the resonances that plagued Anderson's LE.  This allows us to state that the Agmon estimates derived in the previous section are robust to small interaction and hopping, and MBL for low energy states should continue to persist.

We start by splitting the Hamiltonian into three terms corresponding to the hopping, onsite potential, and interaction terms: $\mathcal{H} = T + U + V$.  The landscape is given by
\begin{equation}
    \ket{u} = \frac{1}{U + (T+V)} \ket{1},
\end{equation}
where the suggestive grouping differentiates the nonperturbative term and the perturbations.  We may write this recursively, by utilizing the identity $A^{-1} = B^{-1} + B^{-1}(B-A)A^{-1}$ for $A = U + (T+V)$ and $B = U$:
\begin{equation}
    \ket{u} = \frac{1}{U}\left[I + (T+V)\frac{1}{U + (T+V)}\right] \ket{1},
\end{equation}
and iteration (assuming inverses are well-defined) results in the following infinite series:
\begin{align}
    \ket{u} &= U^{-1} \sum_{n=0}^\infty [(T+V) U^{-1}]^n \ket{1},\nonumber \\ &= U^{-1} \sum_{\beta}\sum_{n=0}^\infty [(T+V) U^{-1}]^n \ket{\beta},
\end{align}
where in the second line we use the fact that $\ket{1}$ is the all ones vector.  The sum is over all basis states in the Fock space (which are eigenstates of $U$).  Next, we call $\mathcal{E}_\alpha = \sum_i n_i^{(\alpha)} \epsilon_i$, which is the unperturbed eigenvalues of $U$.  Inserting $n$ complete sets of eigenstates labelled by $\ket{\alpha_i} \bra{\alpha_i}$, we find:
\begin{align}
    \bra{\alpha}\ket{u} = & \sum_{n=0}^{\infty} \sum_{\alpha_1, \alpha_2, \ldots, \alpha_{n+1}, \beta} \bra{\alpha} U^{-1} \ket{\alpha_1} \times \nonumber\\ & \bra{\alpha_1} (T+V) U^{-1} \ket{\alpha_2} \bra{\alpha_2} (T+V) U^{-1} \ket{\alpha_3} \times \nonumber \\ & \ldots \bra{\alpha_{n}}(T+V) U^{-1}\ket{\alpha_{n+1}}\bra{\alpha_{n+1}}\ket{\beta}.
\end{align}
This can be written in terms of $\mathcal{E}$, resulting in the following expression
\begin{align}
    \bra{\alpha}\ket{u} = \frac{1}{\mathcal{E}_\alpha}\sum_{n=0}^{\infty} \sum_{\alpha_1 = \alpha, \alpha_2, \ldots, \alpha_{n+1}} \frac{\bra{\alpha_1} (T+V) \ket{\alpha_2}}{\mathcal{E}_{\alpha_2}}\times \nonumber\\ \frac{\bra{\alpha_2} (T+V) \ket{\alpha_3}}{{\mathcal{E}_{\alpha_3}}}\ldots \frac{\bra{\alpha_{n}}(T+V) \ket{\alpha_{n+1}}}{{\mathcal{E}_{\alpha_{n+1}}}}.
\end{align}
This series has a graphical representation: we start at a node $\alpha$ and define two propagators: the first carries an amplitude $-\frac{t}{\mathcal{E}_\alpha}$ and occurs along the directed edge from node $\alpha$ to any of its neighboring nodes.  The second is a self-propagator carrying an amplitude $\frac{V}{\mathcal{E}_\alpha}\left(\sum_{\langle i, j\rangle} n_i^{(\alpha)} n_j^{(\alpha)}\right)$.  Thus, the $n$th order term is the sum over all possible paths of products of $n$ propagators along a given path.  

In the conventional LE, quantities like $E - \mathcal{E}_\alpha$ appear in the denominators, which lead to an unregulated perturbation series; however, such resonances are crucially not present in this expansion.  This allows us to show that for sufficiently small perturbation, this series converges:
\begin{theorem}
\label{Th:locator}
 Suppose the onsite potentials have positive support; i.e. $\epsilon_i > \delta$ for $\delta$ a positive constant.  Then the perturbation series for $\bra{\alpha}\ket{u}$ converges for sufficiently small perturbation if $\mathcal{G}_P$ is a cubic lattice in $d$-dimensions.
\end{theorem}
\begin{proof}
The proof is straightforward.  First, note that the eigenstates are not affected by a constant shift in the onsite potential and thus the onsite potential can be chosen to have positive support.  This implies that $\mathcal{E}_\alpha \geq \delta N$ for $\delta$ the minimum value in the support of the onsite potential.  The $n$th order term in the infinite series is given by
\begin{align}
    C^{(\alpha)}_n &= \frac{1}{\mathcal{E}_\alpha} \sum_{\alpha_2, \ldots, \alpha_{n+1}} \frac{\bra{\alpha} (T+V) \ket{\alpha_2}}{\mathcal{E}_{\alpha_2}} \times \nonumber\\ &\frac{\bra{\alpha_2} (T+V) \ket{\alpha_3}}{{\mathcal{E}_{\alpha_3}}}\ldots \frac{\bra{\alpha_{n}}(T+V) \ket{\alpha_{n+1}}}{{\mathcal{E}_{\alpha_{n+1}}}},\nonumber \\
    &\leq \frac{1}{\mathcal{E}_\alpha} \frac{1}{(\delta N)^n} \sum_{\alpha_2, \ldots, \alpha_{n+1}} \bra{\alpha} (T+V)  \ket{\alpha_2} \times \nonumber\\ &\bra{\alpha_2} (T+V)  \ket{\alpha_3}\ldots \bra{\alpha_{n}}(T+V) \ket{\alpha_{n+1}}.
\end{align}
The sum is interpreted as a sum over paths of length $n$ in $\mathcal{G}$ starting at $\alpha$.  Furthermore, we allow self-paths to occur, which result from $\bra{\alpha}V\ket{\alpha}$ being non-zero.  Additionally, since this term can be of order $N$ (since it is equal to $V\sum_{\langle i,j \rangle} n_i^{(\alpha)}n_j^{(\alpha)}$), we create $N$ copies of each node, such that the matrix element due to transitioning to any neighboring node is constant.  For the case of a lattice, the number of paths in this enlarged graph of length $n$ can therefore be bounded by $(cN)^n$, where $c$ is a constant; this follows from the fact that the maximum degree of a node in $\mathcal{G}$ for such graphs is $O(N)$.  The transition amplitude along an edge is upper bounded in absolute value by a constant, which we call $\lambda$.  Therefore, $C^{(\alpha)}_n \leq \frac{1}{\mathcal{E}_\alpha} \frac{\lambda^n (cN)^n}{(\delta N)^n} = \frac{1}{\mathcal{E}_\alpha} \left(\frac{c \lambda}{\delta}\right)^n$.  Choosing $\lambda$ small enough so that $c \lambda/\delta < 1$, $C^{(\alpha)}_n$ decays geometrically in $n$ and the series therefore converges.
\end{proof}

\subsection{Persistence of MBL with Interactions}

While the MBLL illuminates the structure of eigen-states and their exponential decay away from the classically-allowed regions in the Fock space for a given many-body energy, it  does not provide direct information about what the many-body energies are. Of course, at the bare level before performing the locator expansion (i.e., with only a disorder potential present $\left\{ \epsilon_j  \right\}$), the many-body spectrum is known: $\mathcal{E}_\alpha = \sum\limits_{j=1}^L \epsilon_j n_j^{(\alpha)}$ (where $n_j^{(\alpha)}$ is the occupation number of site $j$ in the many-body state $\ket{\alpha}$). The corresponding landscape is just its inverse $\bra{\alpha}\ket{u} \propto \frac{1}{\mathcal{E}_\alpha}$ and in particular, it scales as $1/N$.  The convergence of the landscape under the locator expansion implies that its structure does not change much under perturbation theory, the original localizing traps remain, and the underlying scaling of the effective potential $1/u$ with $N$  persists. However, determining how the many-body energies shift under even a small perturbation is not immediately obvious. Finding self-energy corrections via perturbation theory and in general determining the full spectrum are more difficult problems than perturbing the landscape, because of the resonances complicating the former.  

It is worth mentioning here Mayboroda et al's intriguing results on connecting the spectrum of the Schr{\"o}dinger operator and the structure of the single-particle landscape. As reported in Ref.~[\onlinecite{ll-conj}] and Mayboroda's talk at the 2018 International Congress of Mathematicians~\cite{MayborodaICMTalk}, there is strong evidence that the local minima of the single-particle landscape determine the spectrum to great accuracy. In particular, it was proposed that $E_j \approx \left(1 + \frac{d}{4}\right) \left(\min u^{-1} \right)_j$, which relates the $j$-th energy level to the depth of the $j$-th deepest well in the inverse landscape (here, $d$ is the dimension of space). This heuristic formula was numerically shown to work with unprecedented precision for the lower part of spectrum and some heuristic analytical arguments were presented in Ref.~[\onlinecite{ll-conj}] to support it. However, a rigorous proof supporting these estimates is missing even in the single-particle case at the time of writing of this paper, and it is unclear whether they generalize to discrete lattices or graphs.  However, if there were a one-to-one correspondence between the local minima of the MBLL and many-body energies, this would imply (in conjunction with the convergence of the locator expansion) many-body localization of the  states associated with the MBLL minima (note that there are in general ``more'' states than there are minima in the inverse landscape).

Apart from this observation, a more convincing argument is based on physical reasoning as follows.  Under a weak enough perturbation $\lambda$, the greatest shift we can expect in many-body energies is $O(\lambda N)$.  For $\lambda$ small this shift is small and the size of the classical region does not expand  much from the size in the $t = V = 0$ limit.  Since states are localized at the $t = V = 0$ limit, we conclude that MBL still persists for small $t$ and $V$, at least for the low energy portion of the spectrum.  Utilizing the particle-hole symmetry if $\mathcal{G}_P$ is bipartite also informs us that states are localized at high energies; this is analogous to the argument used in Ref. \cite{ll-dual}.

We note here that by MBL, we imply the existence of a subset of states in the Hilbert space which decay in the Fock space exponentially with the system size in the presence of weak interaction. This is a weaker definition than that used by Basko, Aleiner, and Altshuler~\cite{BAA1}, who called a state $\ket{\Psi_\alpha}$ in the Fock space localized if the creation of a particle-hole pair (or equivalently flipping two spins) results in a state $\sigma^+_i \sigma^-_j\ket{\Psi_\alpha}$ with a support on a finite  -- i.e., $O(1)$ -- number of states in the Fock space. This is contrasted with an ergodic state, whose support is on $O(|\mathscr{H}|)$ states; i.e., spanning most of the Hilbert space. 
Our landscape construction corresponds to an in-between situation where the support (corresponding roughly to a size of classically-allowed regions) may, in principle, be of size $\Omega(1)$, spanning a part of the Hilbert space that scales with system size. Yet, the Basko-Aleiner-Altshuler definition of MBL is weaker than the definition used by Oganesyan and Huse~\cite{Huse2007} and many subsequent publications, where localization of the entire Hilbert space is required. This may occur in a lattice model in one dimension, but most certainly cannot happen in three and higher dimensions, nor in continuous models. Our construction does not rule out the possibility that these stronger versions of MBL hold in a particular model, but can only point to the existence of states exponentially localized in a low-energy part of the Hilbert space (for an arbitrary-dimensional lattice or graph, $\mathcal{G}_P$), where the landscape function is most informative. 

Furthermore, the MBLL does not rule out that states may delocalize at the middle of the band, which may correspond to a percolation-like transition on $\mathcal{G}_F$ as a function of energy.  This potential tendency to delocalize increases as the graph becomes more connected.  This is because choosing an energy towards the middle of the spectrum causes a considerable fraction of nodes to be in the classically allowed region.  Once this fraction is as large as the critical fraction for percolation in $\mathcal{G}_F$, states will become extended throughout the entire Hilbert space.  As $\mathcal{G}_F$ becomes more connected, the percolation transition occurs at a lower filling fraction.  Since $\mathcal{G}_F$ becomes more connected as the dimension increases, this may point to evidence of the formation of a mobility edge for high-dimensional systems.

One caveat to note is when one would expect the perturbative series described in this section to diverge.  This occurs when the physical graph $\mathcal{G}_P$ is highly connected.  Note that if $\mathcal{G}_P$ is locally connected (for instance in a lattice), then $\mathcal{G}_F$ has a maximum degree that is linear in $N$.  The moment $\mathcal{G}_P$ becomes more strongly connected by longer range hopping, $\mathcal{G}_F$ will have nodes with degree that scale as $\Omega(N)$, and the resulting bound presented in Theorem~\ref{Th:locator} of this section will not hold.  This is consistent with the intuition that systems with long-range hopping do not exhibit MBL.

\section{Conclusions}

We construct a generalization of the single-particle localization landscape to interacting many-body systems. The many-body localization landscape is a useful representation of interacting many-body systems on a general class of lattices.  The MBLL is given by $(\mathcal{H}+K) \ket{u} = \ket{1}$ in Fock space, and thus is more appropriately defined with respect to a graph in the Fock space, which we call $\mathcal{G}_F$.  We show that the MBLL satisfies $|\psi_\alpha| \leq E (\max_\beta |\psi_\beta|) u_\alpha$, and thus the valleys of $1/u$ can be interpreted as ``traps'' (classically allowed regions) with respect to $\mathcal{G}_F$ that confine many-body eigenstates.

We then describe the behavior of eigenfunctions outside the classically allowed regions by generalizing Agmon decay estimates.  These decay estimates are defined with respect to paths in the Fock-space graph $\mathcal{G}_F$ and can be interpreted as the minimization of an effective action.  We show that the eigenfunctions begin to decay slowly once they pass through the classically forbidden region, but it is only after roughly $\Theta(N^{1/2})$ steps into the forbidden region that the decay rate begins to scale with the system size.  When the depth into the forbidden region is $\Theta(N)$, the decay rate becomes exponential in the path length.

A locality theorem is then proved, which roughly states that the eigenstate in a particular well, as defined by the MBLL, is sufficiently close to the span of eigenstates in the entire Fock space whose energies are sufficiently close; a similar statement about global eigenstates sufficiently close to the span of well eigenstates with close energies is also proved.  This implies that eigenstates tend to localize in a small number of individual wells as opposed to being spread out in the Hilbert space among many wells.

We then develop a perturbative expansion for the MBLL similar to the  locator expansion; i.e. we treat the disorder as the leading term in the Hamiltonian and perform a series expansion in both the hopping and interaction.  Unlike in the locator expansion, where one encounters resonances, we show that the perturbative expansion for the MBLL converges for small hopping and interaction.  This implies that the landscape only changes by a multiplicative constant from the unperturbed landscape.  Assuming the corresponding energy shift is not too large, the size of the classically allowed regions should not change by much with small interaction and the localization properties should be similar to those of the non-interacting problem. 

The convergent locator expansion combined with the locality theorem allows us to prove that at least for particular kinds of disorder distributions, localization does persist in the presence of weak interactions for a part of the Hilbert space. Namely, we can choose an initial disorder distribution with a desired level statistics, such that the wells do not percolate up to an energy threshold.  The corresponding landscape without hopping and interactions is just the disorder potential itself and the convergent locator series implies that it is weakly corrected as long as the disorder potential is dominant compared to the hopping and interactions. The locality theorem implies that the set of tightly localized states in the disorder-only model remain exponentially localized in a small number of non-percolating wells in the presence of small interactions. This rigorously establishes a weak version of many-body localization for a class of models.

We note that the landscape construction crucially does not shed light on whether full MBL -- by which we mean localization of all states in the Hilbert space and the complete breakdown of ergodicity -- may exist in a model. What it does show, however, is that there exists a rich variety of possible localization structures in the Fock space, which include the possibility  of Basko-Aleiner-Altshuler-type MBL, full MBL of the entire spectrum, and a continuum of structures in between. In the language of the MBLL, there is a question of quantifying the size of the ``classically allowed regions,''  where the many-body wave-functions do not decay; this is analogous to the size of support in Fock space of a state with a particle-hole excitation as described in Ref.~\cite{BAA1}). This region can consist of a finite number of states~\cite{BAA1}, or scale with the system size while still covering a measure zero set in the entire Hilbert space.  Furthermore, the decay properties beyond these regions may differ -- e.g. power-law vs exponential -- based on the region in the Fock space graph.  As we have described throughout this manuscript, different parts of the Hilbert space may in principle have different properties.  This points to a variety of possible phases and suggests that a finer classification of MBL needs to be developed to describe these structures. 

We also point out some limitations of the MBLL construction we presented in this paper.  For one, we cannot fully determine the energy shifts due to interaction, which is crucial to determining the size of classically allowed regions.  This computation requires one to calculate self-energies, a method which is only known to be rigorous on a Bethe lattice.  However, Mayboroda and collaborators have several conjectures (numerically verified) stating the rather astounding notion that the set of local minima of the inverse landscape can determine the spectrum to high accuracy \cite{ll-conj}.  In particular, if such a result can be proved, we would not only be able to approximate energy shifts, but also determine level statistics of the spectrum, which are commonly used in identifying the MBL transition.

Secondly, it is not known whether the Agmon estimates derived are tight -- it may be possible to improve these bounds by eliminating or refining the presence of $\sqrt{1/|N(\alpha)|}$ in the ``graph momentum.''  Doing so would allow one to make stronger statements of rapid exponential decay for a larger fraction of states in the Fock space, which would improve understanding of the mechanism for MBL.


In the perturbative series we performed, where we treat the disorder as large, the landscape essentially only contains the onsite disorder potential. The landscape thus gives us more information when the hopping becomes considerably large.  Moreover, there are rigorous mathematical proofs for the existence of Anderson localization in 1D non-interacting systems~\cite{math-anderson-1,math-anderson-2}.  Thus, a more natural way to use perturbation theory would be to consider the non-interacting Anderson model as the starting point and calculate the MBLL perturbatively in the interaction only.  

On top of these avenues for exploration, there are several other modifications that may also lead to interesting effects.  For one, we argued that the sign of $t$ should not matter due to particle-hole symmetry if the lattice is bipartite, but it would be important to understand whether this constraint can be relaxed to frustrated lattices that are not bipartite. It is conceivable that frustrated disordered lattices may favor delocalization -- an interesting scenario. Secondly, we assumed in our treatment that the interaction term is repulsive, but a similar theory can be developed for attractive interactions. Such ``negative-$U$'' Hubbard models with disorder are of great interest for understanding the superconductor-to-insulator transition in disordered films, where a crossover is expected from a superconductor to a regime of localized Cooper pairs to that of the more conventional electronic insulator~\cite{AndersonBosons}. Thirdly, we briefly mentioned that other non-local interactions can be used, so long as the interaction strength is negative (or there exists a similar particle-hole symmetry so that the sign of the interaction does not matter) -- it would similarly be important to understand for what class of interactions the landscape picture would break down.  In particular, it may be possible to develop the MBLL theory in the case of the Dirac Sachdev-Ye-Kitaev (SYK) model~\cite{SY}, in which case the MBLL may contain signatures of ergodicity.  

The current theory developed shows that the many-body localization landscape is a relatively simple yet powerful object that reveals the rich nature of many-body localization and possibly other interacting many-body problems. Furthermore, the computational price for calculating the many-body localization landscape is smaller than that required for exact diagonalization; therefore the exact, non-perturbative landscape can be computed numerically for larger system sizes than what has been possible with exact diagonalization.

\acknowledgements{The authors are grateful to Yi-Ting Hsu for useful discussions and a fruitful collaboration on a related numerical project to be published elsewhere. S.B. was supported by NSF DMR-1613029 while at the Physics Frontier Center in JQI. This research was also supported by US-ARO (contract No. W911NF1310172) and DARPA DRINQS program (Y.L.), DOE- BES (DESC0001911) (V.G.) and the Simons Foundation. V.G. is especially grateful to the Simons Foundation for hospitality during the MPS Annual Meetings, where this research project was conceived.}

\bibliography{MBLL}

\end{document}